\documentclass[]{interact}

\usepackage{epstopdf}% To incorporate .eps illustrations using PDFLaTeX, etc.
\usepackage[caption=false]{subfig}% Support for small, `sub' figures and tables
\usepackage[numbers,sort&compress]{natbib}% Citation support using natbib.sty
\bibpunct[, ]{[}{]}{,}{n}{,}{,}% Citation support using natbib.sty
\makeatletter% @ becomes a letter
\def\NAT@def@citea{\def\@citea{\NAT@separator}}% Suppress spaces between citations using natbib.sty
\makeatother% @ becomes a symbol again

\usepackage{amsthm}%定理和证明类环境的配置
\usepackage[usenames]{color}
\usepackage{xspace,colortbl}
\usepackage{latexsym,url,amscd,amssymb}
\usepackage{amsmath}
\usepackage{hyperref}
\usepackage{cases}
\usepackage{caption}
\usepackage{algorithm}
\usepackage{algpseudocode}
\usepackage{listings}
\usepackage{rotating}
\usepackage{framed}
\usepackage{amssymb}

\usepackage{amsmath}\numberwithin{equation}{section}
\usepackage{multirow}
\usepackage{makecell}
\usepackage{booktabs}

\DeclareMathOperator*{\argmin}{arg\,min}

\def\A{{\mathcal A}}
\def\B{{\mathcal B}}

\def\I{{\mathcal I}}

\def\N{{\mathcal N}}

\def\S{{\mathcal S}}
\def\T{{\mathcal T}}

\newtheorem{lemma}{Lemma}[section]

\newtheorem{theorem}{Theorem}[section]
\newtheorem{assumption}{Assumption}[section]

\renewcommand{\thesection}{\arabic{section}}
\renewcommand{\thesubsection}{\arabic{section}.\arabic{subsection}}

\begin{document}

\title{A Primal Dual Active Set with Continuation Algorithm for $\ell_0$-Penalized High-dimensional Accelerated Failure Time Model}
\author{
\name{Peili Li\textsuperscript{a,b}, Ruoying Hu\textsuperscript{a}, Yanyun Ding\textsuperscript{c}, and Yunhai Xiao\textsuperscript{b}\thanks{CONTACT Yunhai Xiao. Email: yhxiao@henu.edu.cn}}
\affil{\textsuperscript{a}School of Mathematics and Statistics, Henan University, Kaifeng 475000, P.R. China; \\
\textsuperscript{b}Center for Applied Mathematics of Henan Province, Henan University, Zhengzhou 450046, P.R. China; \\
\textsuperscript{c}Institute of Applied Mathematics, Shenzhen Polytechnic University, Shenzhen 518055, P.R. China.}
}

%dingyanyun@szpu.edu.cn
\maketitle

\begin{abstract}
The accelerated failure time model has garnered attention due to its intuitive linear regression interpretation and has been successfully applied in fields such as biostatistics, clinical medicine, economics, and social sciences.
This paper considers a weighted least squares estimation method with an $\ell_0$-penalty based on right-censored data in a high-dimensional setting.
For practical implementation, we adopt an efficient primal dual active set algorithm and utilize a continuous strategy to select the appropriate regularization parameter.
By employing the mutual incoherence property and restricted isometry property of the covariate matrix, we perform an error analysis for the estimated variables in the active set during the iteration process.
Furthermore, we identify a distinctive monotonicity in the active set and show that the algorithm terminates at the oracle solution in a finite number of steps. Finally, we perform extensive numerical experiments using both simulated data and real breast cancer datasets to assess the performance benefits of our method in comparison to other existing approaches.
\end{abstract}


%%%%%%%%%%%%%%%%%%%%%%%%%%%%%%%%%%%%%%%%%%%%%%%%%%%%%%
\setcounter{equation}{0}
\section{Introduction}\label{section1}
%%%%%%%%%%%%%%%%%%%%%%%%%%%%%%%%%%%%%%%%%%%%%%%%%%%%%%

In research fields, such as mathematics, statistics, biology, and economics, survival analysis is commonly used to explore the mechanisms behind event occurrence or failure.
In these fields, it often involves investigating various factors influencing the subjects to identify the key ones, which allows us to assess and predict when events are likely to happen.
In recent years, a variety of survival models have been introduced, including the Bayesian parametric survival model, Cox proportional hazards model, accelerated failure time (AFT) model, and support vector machine model.
Among these models, the AFT model establishes a linear relationship between the logarithm or a known monotonic transformation of the failure time and the covariates. This linear regression approach contributes to its broader applicability compared to other models.
Mathematically, the AFT model can be expressed as
\begin{equation}\label{mmodel}
	\text{ln}(T_i)=X_i^{\top}\beta^*+\epsilon_i, \quad i=1,\ldots,n,
\end{equation}
where $T_i$ is the failure time, $X_i \in \mathbb{R}^{p}$ is the covariate vector, $\beta^*\in \mathbb{R}^{p}$ is the regression coefficient, and $\epsilon_i$'s are the independent and identically distributed (i.i.d.) random errors.
Here, we assume that $T_i$ is right-censored, meaning the starting time of observation is known, but the event termination time is unknown, which prevents us from determining the exact survival time. But we only know that the survival time is greater than the observation time.
At this point, we can obtain the observed data $\{(Y_i,\delta_i,X_i)\}_{i=1}^n$, where $Y_i=\text{min}\{\text{ln}(T_i),\text{ln}(C_i)\}$, $C_i$ is the censoring time and $\delta_i=\textbf{1}_{\left\{T_i \leq C_i\right\}}$ is the censoring indicator.

It is well known that   there are several methods available to estimate the coefficient $\beta^*$ appeared in \eqref{mmodel}, including the Buckley-James method \cite{BJ1979}, rank-based method \cite{Y1993} and weighted least-squares method \cite{S1996, SW1993}.
In this paper, we primarily focus on the estimation and variable selection in high-dimensional settings, where the sample size $n$ is assumed to be much smaller than the dimension $p$. In such cases, sparse penalty is often employed, which use a penalty term into the traditional estimation loss function.
For example, Johnson et al. \cite{J2008, JLZ2008} considered smoothly clipped absolute deviation penalty \cite{FL2001} to the Buckley-James estimator.
Cai et al. \cite{CHT2009} introduced an adaptive lasso method \cite{Z2006} based on rank-based estimation.
Hu et al. \cite{HC2013}, Huang et al. \cite{HM2010,HMX2006}, and Khan et al. \cite{KS2016} extended lasso \cite{T1996}, adaptive weighted elastic net \cite{HZ2010, ZZ2009}, bridge penalty \cite{FF1993} and minimax concave penalty \cite{Z2010} to weighted least-squares estimation, respectively.
Specifically, Cheng et al. \cite{CFHJZ2022} extended the non-convex, non-smooth $\ell_0$-penalty to weighted least-squares estimation. The $\ell_0$-penalty is well-suited for sparse estimation, as it effectively mitigates overfitting and enhances interpretability.
Additionally, Cheng et al. \cite{CFHJZ2022} extended the support detection and root finding algorithm \cite{HJLL2018} and derived $\ell_\infty$-error bounds for the solution sequence. Under specific conditions, they demonstrated that the estimated support accurately includes the true support within a finite number of steps.

In this paper, we focus on the high-dimensional AFT model with an $\ell_0$-penalty. We observe that the challenges presented by high dimensionality, heavy censoring, and non-convex, non-smooth characteristics greatly complicate the model estimation and variable selection.
For practical implementation, in this paper, we extend the efficient primal dual active set with continuation (PDASC) algorithm proposed by Jiao et al. \cite{JJL2015}, referred to as AFT-PDASC.
In each iteration, the active and inactive sets are initially identified based on the primal and dual variables from the previous iteration.
The primal variable is updated by solving a least-squares problem for the active set, while the dual variable is updated explicitly using the KKT conditions.
Furthermore, a continuation strategy for the regularization parameter is incorporated to achieve algorithm's convergence.
We observe that the active set exhibits a special monotonicity per iteration, where it gradually grows and eventually matches the true active set as the regularization parameter decreases.
In contrast to the work of Cheng et al. \cite{CFHJZ2022}, we analyze the $\ell_2$- and $\ell_\infty$- error bounds of the solution sequence on the active set, assuming the mutual incoherence property and restricted isometry property  of the covariate matrix.
In addition, we prove that the iterative process of this algorithm terminates in a finite number of steps at the oracle solution.
Finally, we perform simulation experiments to evaluate the estimation performance of AFT-PDASC for the high-dimensional AFT model, evaluate the effects of the involved parameters' values, and present performance comparisons with other existing algorithms using simulated data and real-world datasets.

\textbf{Organization:} The remainder of this paper is organized as follows. In Section \ref{algo}, we present the $\ell_0$-penalized AFT model, then transform it into a standard least squares estimation model through some simple algebraic manipulations, and finally list the iterative framework of AFT-PDASC.
In Section \ref{the}, we present some relevant theoretical results, which includes error analysis, the monotonicity analysis of the active set, and the finite-step termination property of AFT-PDASC.
In Section \ref{num}, we design some simulation and real-world experiments to evaluate the numerical performance of AFT-PDASC.
Finally, in Section \ref{con}, we conclude this paper.

\textbf{Notation:} For a vector $\beta \in \mathbb{R}^{p}$, we denote $\|\beta\|:=\sqrt{\sum_{i=1}^{p}\beta_i^2}$ and $\|\beta\|_{\infty}:=\max_{1\leq i \leq p}\{|\beta_i|\}$. The symbol $\|\beta\|_{T,\infty}$ is the $T$-th largest elements (in absolute value), $\|\beta\|_0$ denotes the number of non-zero elements in $\beta$ and $\text{supp}(\beta):=\{i \ | \ \beta_i\neq 0, i=1,2,\ldots,p\}$. For an index set $\A$, we denote $|\A|$ as its length and $\A^c$ is the complement. In addition, we also denote $\beta_{\A}:=\{\beta_i, i \in \A\}\in \mathbb{R}^{|\A|}$.
For a matrix $X\in \mathbb{R}^{n\times p}$, we use $\|X\|_{\infty}$ to represent its maximum value (in absolute value) and use the notation $X_{\A}\in \mathbb{R}^{n\times|\A|}$.
We use $\beta^*$ and $\hat{\beta}$ to denote the true and the estimated regression coefficient, respectively. We assume the true coefficient $\beta^*$ has $K$ non-zero components along with its support set denoted as $\A^*$, i.e., $|\A^*|=K$. Finally, we denote $\S=:\{1,2,\ldots,p\}$ and let $\I^*:=\S \setminus \A^*$.
%%%%%%%%%%%%%%%%%%%%%%%%%%%%%%%%%%%%%%%%%%%%%%%55
\section{AFT-PDASC algorithm for $\ell_0$-penalized AFT model}\label{algo}
\setcounter{equation}{0}
%%%%%%%%%%%%%%%%%%%%%%%%%%%%%%%%%%%%%%%%%%%%%%%%%%%%
\subsection{$\ell_0$-penalized AFT model}
%%%%%%%%%%%%%%%%%%%%%%%%%%%%%%%%

Let $Y_{(1)}\leq \ldots \leq Y_{(n)}$ be the order statistics of $Y_i$'s, $\delta_{(1)}, \ldots, \delta_{(n)}$ be the associated censoring indicators, and $X_{(1)}, \ldots, X_{(n)}$ be the associated covariates.
The  weighted least-squares estimation method with an $\ell_0$-penalty for  AFT model \eqref{mmodel} is given by:
\begin{equation}\label{2model1}
\min\limits_{\beta\in \mathbb{R}^p}\frac{1}{2} \sum_{i=1}^n w_{(i)}\big(Y_{(i)}-X_{(i)}^{\top} \beta\big)^2 +\lambda\|\beta\|_0,
\end{equation}
where $w_{(i)}$'s  are the jumps in Kaplan-Meier estimator in the form of
$w_{(1)}=\delta_{(1)}/{n}$ and that
$$
w_{(i)}=\frac{\delta_{(i)}}{n-i+1} \prod_{j=1}^{i-1}\big(\frac{n-j}{n-j+1}\big)^{\delta_{(j)}},\quad i=2, \ldots, n,
$$
and $\lambda>0$ is a regularization parameter to control the sparsity level of $\beta$.
For convenience, we make a standardization on  (\ref{2model1}).
Define
$$
\tilde{Y}:=\Big(\sqrt{w_{(1)}}Y_{(1)}, \ldots, \sqrt{w_{(n)}}Y_{(n)}\Big)^{\top} \quad\text{and}\quad \tilde{X}:=\Big(\sqrt{w_{(1)}}X_{(1)}, \ldots, \sqrt{w_{(n)}}X_{(n)}\Big)^{\top}.
$$
Assume that $\|\tilde{X}_i\|_2>0$, where $\tilde{X}_i\in\mathbb{R}^n$ is the $i$-th column of $\tilde{X}$.
Define
$$
D:=\text{diag}\Big(\frac{1}{\|\tilde{X}_1\|_2},\ldots,\frac{1}{\|\tilde{X}_p\|_2}\Big),
$$
and denote
$\eta:=D^{-1}\beta$, $\bar{X}:=\tilde{X}D$, $\bar{Y}:=\tilde{Y}$. Furthermore, let $\bar{Y} := \bar{X}\eta^* + \gamma$, where $\gamma$ is an error satisfying $\|\gamma\|\leq \bar{\epsilon}$, and $\bar{\epsilon}\geq 0$ is a noise level.
Using these notation, we can rewrite \eqref{2model1} as the following $\ell_0$-penalized least-squares problem:
\begin{equation}\label{2model2}
\min\limits_{\eta\in \mathbb{R}^p}L(\eta):=\frac{1}{2} \|\bar{Y}-\bar{X} \eta\|^2 +\lambda\|\eta\|_0.
\end{equation}
In the following subsection, we will develop an algorithm to solve (\ref{2model2}) to obtain the solution $\hat{\eta}$. Subsequently,  the estimated coefficient $\hat{\beta}$ of \eqref{mmodel} can be computed by using $\hat{\beta} = D\hat{\eta}$.

%%%%%%%%%%%%%%%%%%%%%%%%%%%%%%%%%%%%%%%%%%%
\subsection{AFT-PDASC algorithm}
%%%%%%%%%%%%%%%%%%%%%%%%%%%%%%%%%
Due to the non-convex and non-smooth nature of the objective function in (\ref{2model2}), we focus on a coordinate-wise minimization approach.
A vector $\eta=(\eta_1,\ldots,\eta_p)^{\top}\in \mathbb{R}^p$ is said to be a coordinate-wise minimizer of $L(\eta)$ if it is a local minimizer along each coordinate direction, that is,
$$
\eta_i\in\argmin _{t \in \mathbb{R}} L\big(\eta_1,...,\eta_{i-1},t,\eta_{i+1},...,\eta_p\big).
$$
The following lemma presents the KKT conditions for the coordinate-wise minimizers of $ L(\eta)$, which plays a crucial role in the development of our subsequent algorithm. For the proof, one can refer to \cite{JJLR2013, LJLK2022}.

\begin{lemma}\label{kktc}
 If $\bar{\eta}$ is a minimizer of (\ref{2model2}), then there exists a $\bar{d}$ such that the following KKT system holds:
 \begin{equation}\label{kk}
\left\{
\begin{aligned}
&\bar{d}=\bar{X}^{\top}(\bar{Y}-\bar{X}\bar{\eta}),\\
&\bar{\eta}=S_{\lambda}(\bar{\eta}+\bar{d}),
\end{aligned}
\right.
\end{equation}
where the $i$-th element of $S_{\lambda}(x)$ is defined as
\begin{equation}\label{sl}
	\Big(S_{\lambda}(x)\Big)_{i}
	\begin{cases}
		=0, &if \ |x_{i}|<\sqrt{2\lambda},\\
		\in\{0,x_{i}\},&if \ |x_{i}|=\sqrt{2\lambda},\\
		=x_{i},&if \ |x_{i}|>\sqrt{2\lambda}.
	\end{cases}
\end{equation}
Conversely, if $\bar{\eta}$ and $\bar{d}$ satisfy (\ref{kk}), then $\bar{\eta}$ is a local minimizer of model (\ref{2model2}).
\end{lemma}

Let $\bar{\A}:=\text{supp}(\bar{\eta})$ and $\bar{\I}:=\bar{\A}^{c}$. Thus, it can be seen from (\ref{kk}) and (\ref{sl}) that
\begin{equation}\label{calai}
\bar{\A}:=\Big\{i : |\bar{\eta}_i+\bar{d}_i|> \sqrt{2\lambda}\Big\}, \quad \bar{\I}:=\Big\{i : |\bar{\eta}_i+\bar{d}_i|\leq \sqrt{2\lambda}\Big\},
\end{equation}
and that
\begin{equation}\label{upai}
\bar{\eta}_{\bar{\I}}=0, \quad \bar{\eta}_{\bar{\A}}=\Big(\bar{X}_{\bar{\A}}^{\top}\bar{X}_{\bar{\A}}\Big)^{-1}\bar{X}_{\bar{\A}}^{\top}\bar{Y}, \quad \bar{d}_{\bar{\A}}=0, \quad \bar{d}_{\bar{\I}}=\bar{X}^{\top}_{\bar{\I}}\Big(\bar{Y}-\bar{X}_{\bar{\A}}\bar{\eta}_{\bar{\A}}\Big).
\end{equation}
Based on \eqref{calai} and \eqref{upai}, we can iteratively derive $\bar\eta$ and $\bar d$ by applying a specific rule.
Typically, for the active set $\mathcal{A}^{k}$, we can compute the primal variable $\eta^{k}$ by solving a least-squares problem, and then explicitly update the dual variable $d^{k}$ using the KKT conditions. The latest $\beta^{k}$ can also be get using the relation $\beta^{k}=D\eta^{k}$.
This process is repeated until the active sets in successive steps become consistent, or the maximum number of iterations is reached.

It is important to note that the primal-dual active set method is essentially equivalent to the well-known semismooth Newton method, which demonstrates local superlinear convergence, see \cite{2002primal} for more details.
However, to fully leverage this characteristic, a good initial guess is often required.
In this paper, we employ a continuation technique for the regularization parameter $\lambda$ to address this issue.
Specifically, given an initial value $\lambda_0$, we set $\lambda_k = \lambda_0\rho^k$ with $\rho\in (0,1)$. Based on $\lambda_k$, we iterate the above inner loop to obtain $\eta(\lambda_k)$ and $d(\lambda_k)$. Next, we use $\eta(\lambda_k)$ and $d(\lambda_k)$ as the initial values for the corresponding problem with $\lambda_{k+1}$ and continue the iterative process until a discrete analogue of the discrepancy principle is satisfied.

Based on the above analysis, we outline the iterative framework of the AFT-PDASC below.

\begin{framed}
	\noindent
	{\bf Algorithm $1$: AFT-PDASC}
	\vskip 2.0mm \hrule \vskip 4mm
	\noindent
{\bf 1}. Choose $\lambda_0\geq\frac{1}{2}\|\bar{X}^{\top}\bar{Y}\|_{\infty}^2$ and define $\A(\lambda_0):=\varnothing$ and $\eta(\lambda_0):=0$. Let $d(\lambda_0):=\bar{X}^{\top}\bar{Y}$ and choose $\rho\in (0,1)$ and $\bar{\epsilon}>0$. Choose positive integers $K_{\max}\in \N$ and $J_{\max}\in \N$.\\
{\bf 2.} \textbf{for $k=1,2,\ldots,K_{\max}$ do}\\
{\bf 3.} \qquad Let $\lambda_k:=\rho \lambda_{k-1}$, $\A_0:=\A(\lambda_{k-1})$, and $(\eta^0,d^0):=(\eta(\lambda_{k-1}),d(\lambda_{k-1}))$.\\
{\bf 4.} \qquad \textbf{for $j=1,2,\ldots,J_{\max}$ do}\\
{\bf 5.} \qquad \qquad Compute
$\A_j=\{i : |\eta_i^{j-1}+d_i^{j-1}|>\sqrt{2\lambda_k}\}$ and let $\I_j=\A_j^c$;\\
{\bf 6.} \qquad \qquad Break if $\A_j=\A_{j-1}$;\\
{\bf 7.} \qquad \qquad Otherwise, compute
\begin{align*}
&\eta^j_{\I_j}:=0, \quad \eta^j_{\A_j}:=(\bar{X}_{\A_j}^{\top}\bar{X}_{\A_j})^{-1}\bar{X}_{\A_j}^{\top}\bar{Y}, \\
&d^j_{\A_j}:=0, \quad d^j_{\I_j}:=\bar{X}^{\top}_{\I_j}(\bar{Y}-\bar{X}_{\A_j}\eta^j_{\A_j}).
\end{align*}
{\bf 8.} \qquad\qquad Compute  $\beta^j:=D\eta^j$.\\
{\bf 9.} \qquad\textbf{end for}\\
{\bf 10.} \quad \ Choose $\tilde{j}=\min(J_{\max},j)$. Compute
$$
\A(\lambda_k):=\big\{i : |\eta_i^{\tilde{j}}+d_i^{\tilde{j}}|>\sqrt{2\lambda_k}\big\} \ \text{and} \ \big(\eta(\lambda_{k}),d(\lambda_{k})\big):=(\eta^{\tilde{j}},d^{\tilde{j}}).
$$
{\bf 11.} \quad \  Compute $\beta(\lambda_k):=D\eta(\lambda_k)$.\\
{\bf 12.} \quad \ Stop if $\|\bar{Y}-\bar{X}\eta(\lambda_{k})\|\leq \bar{\epsilon}$. Otherwise, continue.\\
{\bf 13.} \textbf{end for}
\end{framed}

It is important to highlight that Cheng et al. \cite{CFHJZ2022} also proposed an algorithm based on the KKT system for the coordinate-wise minimizers of model (\ref{2model2}), but there are several key differences compared to AFT-PDASC.
In computing the active set, Cheng et al. used $\|\eta + d\|_{T, \infty}$ to eliminate the parameter $\lambda$, but this approach relies on a new parameter $T$.
Additionally, a step size $\tau$ is also required for computing the active set, which makes both the theoretical analysis and numerical performance highly dependent on this step size.
In summary, the performance of the algorithm proposed by Cheng et al. \cite{CFHJZ2022} heavily depends on the selection of the parameters $T$ and $\tau$, which can be a challenging task in practical implementation.
In contrast, our approach directly uses the regularization parameter \(\lambda\) and employs a continuation technique to allow automatic selection of this parameter.
Furthermore, the theoretical analysis presented in this paper differs considerably from the work of Cheng et al. \cite{CFHJZ2022}.
Most importantly, the numerical comparisons with \cite{CFHJZ2022} demonstrate that the algorithm presented in this paper can consistently achieve high-precision solutions generally.

%%%%%%%%%%%%%%%%%%%%%%%%%%%%%%%%%%%%%%%%%%%%%%%%%%%%%
\section{Theoretical analysis}\label{the}
\setcounter{equation}{0}
%%%%%%%%%%%%%%%%%%%%%%%%%%%%%%%%%%%%%%%%%%%%%%%%%%%%%%

In this section, we derive an error analysis for the estimated coefficients on the active set during the iterative process, leveraging the mutual incoherence property (MIP) and the restricted isometry property (RIP) of the covariate matrix \(\bar{X}\). Additionally, we provide a finite-step termination analysis of AFT-PDASC at the oracle solution by using the special monotonicity of the active set.
For simplicity, we let $\lambda > 0$ and $s > 0$, and define $\mathcal{T}_{\lambda,s} := \{i : |\eta^*_i| \geq \sqrt{2 \lambda} s\}$ to identify the indices of $\eta^*$ where the values are relatively large.
In addition, we also define the least-squares solution on the true support set $\A^*$ as the oracle solution $\eta^o\in \mathbb{R}^{|\A^*|}$, which obeys the form $\eta^o:=(\bar{X}_{\A^*}^{\top}\bar{X}_{\A^*})^{-1}\bar{X}_{\A^*}^{\top}\bar{Y}$.

For the purpose of the subsequent theoretical analysis, we need a couple of assumptions on the covariate matrix $\bar{X}$:

\begin{assumption}\label{assum}
	The following conditions hold:
\begin{itemize}
  \item[(1)] The mutual coherence $\nu=\max\limits _{1 \leq i\neq j \leq p}\Big\{|\bar{X}_i^{\top} \bar{X}_j|\Big\}$ of the covariate matrix $\bar{X}$ is small, where $\bar{X}_i$ is the $i$-th column of $\bar{X}$.
  \item[(2)] There exists a constant $\delta\in (0,1)$ such that $(1-\delta)\|\eta\|^2\leq\|\bar{X}\eta\|^2\leq(1+\delta)\|\eta\|^2$ for any $\eta \in \mathbb{R}^{p}$ satisfying $\|\eta\|_0\leq s$. Here, we define $\delta_s$ as the infimum of all
parameters $\delta$ for which the RIP holds.
\end{itemize}
\end{assumption}

The following lemma provides some basic estimates under the MIP and RIP conditions.
One may refer to \cite{DM2009,NT2009,JJL2015} for its proof.
\begin{lemma}\label{mipl}
Assume $\A$ and $\B$ are disjoint subsets of $\S$. Then
\begin{enumerate}
  \item
\begin{align*}
\|\bar{X}_{\A}^{\top}\bar{Y}\|_{\infty}\leq\|\bar{Y}\|, &\quad \|\bar{X}_{\B}^{\top}\bar{X}_{\A}\eta_{\A}\|_{\infty}\leq \nu|\A|\|\eta_{\A}\|_{\infty}, \\ \|(\bar{X}_{\A}^{\top}\bar{X}_{\A})^{-1}\eta_{\A}\|_{\infty}\leq & \frac{\|\eta_{\A}\|_{\infty}}{1-(|\A|-1)\nu}, \ if \ (|\A|-1)\nu<1.
\end{align*}
  \item
  \begin{align*}
  \|\bar{X}_{\A}^{\top}\bar{X}_{\A}\eta_{\A}\|\gtreqless(1 \mp \delta_{|\A|})\|\eta_{\A}\|, &\quad
  \|(\bar{X}_{\A}^{\top}\bar{X}_{\A})^{-1}\eta_{\A}\|\gtreqless \frac{\|\eta_{\A}\|}{1 \pm\delta_{|\A|}},\\
  \|\bar{X}_{\A}^{\top}\bar{X}_{\B}\|\leq \delta_{|\A|+|\B|}, \quad  \|(\bar{X}_{\A}^{\top}\bar{X}_{\A})^{-1}&\bar{X}_{\A}^{\top}\bar{Y}\|\leq \frac{\|\bar{Y}\|}{\sqrt{1-\delta_{|\A|}}}, \quad \delta_s \leq \delta_{s'}, \ if \ s<s'.
  \end{align*}
\end{enumerate}
\end{lemma}

For the subsequent theoretical analysis, we also require the following assumption on the noise level $\bar{\epsilon}$.
\begin{assumption}\label{assum2}
The noise level $\bar{\epsilon}$ is small in sense that $\bar{\epsilon}\le\alpha\min_{i\in \A^*}\{|\eta_i^*|\}$ for a certain one $0\le\alpha<1/2$.
\end{assumption}

We  now provide the main theoretical results of the estimation method \eqref{2model2}  based on the MIP condition in Assumption \ref{assum} (1) on the covariate matrix $\bar{X}$.

\begin{theorem}\label{th1}
Suppose that Assumption \ref{assum2} holds. If $\nu<\frac{1-2\alpha}{2K-1}$, then for any $\rho \in\big(((2K-1) \nu+2\alpha)^2, 1\big)$, it holds that:
\begin{itemize}
\item[(1)] In the $\lambda_k$ subproblem, for $j=1,2,\ldots,J_{\max}$, we set $\B:=\A^* \backslash \A_j $ and $\I_j:=\S \backslash \A_j$. If $|\A_j|\le K$, then we have
 $$
\begin{aligned}
	\big\|\eta^{j}_{\A_j} -\eta_{\A_j}^* \big\|_{\infty} & \leq \frac{1}{1-(|\A_j|-1) \nu}\Big(\nu|\B|\cdot\|\eta_{\B}^*\|_{\infty}+\bar{\epsilon}\Big), \\
	\big\|\beta^j_{\A_j} -\beta_{\A_j}^*\big\|_{\infty}
&\le \frac{\|D\|_{\infty}}{1-(|\A_j|-1)\nu}\Big(\nu|\B|\cdot\|D^{-1}\|_{\infty}\cdot\|\beta^*_{\B}\|_{\infty}+\bar{\epsilon}\Big), \\
\big\|\beta^j_{\A_j} -\beta_{\A_j}^*\big\|
&\le \frac{\sqrt{K}\|D\|_{\infty}}{1-(|\A_j|-1)\nu}\Big(\nu|\B|\cdot\|D^{-1}\|_{\infty}\cdot\|\beta^*_{\B}\|_{\infty}+\bar{\epsilon}\Big).
\end{aligned}
$$
\item[(2)] For $k=0,1,2,\ldots,K_{\max}$, there exist $s_1, s_2 \in\big(1/(1-K \nu+\nu-\alpha), 1/(K \nu+\alpha)\big)$ with $s_1>s_2$, such that the active sets ${\A(\lambda_k)}$ have the special monotonicity, that is,
\begin{equation}\label{col22}
\text{If} \ \T_{\lambda_k,s_1}\subseteq \A(\lambda_{k-1})\subseteq \A^*, \quad \text{then} \ \T_{\lambda_k,s_2}\subseteq \A(\lambda_{k})\subseteq \A^*.
\end{equation}
Furthermore, AFT-PDASC  terminates in a finite number of steps.

\item[(3)] Let $\alpha\le\frac{1-2(K-1)\nu}{K+3}$ and let $\xi:=\frac{1-2(K-1) \nu-2 \alpha-\alpha^2}{2 K} \min _{i \in \A^*}\{|\eta_i^*|^2\}$. Then, for any $\lambda \in\big(\frac{\bar{\epsilon}^2}{2}, \xi\big)$, AFT-PDASC terminates at the oracle solution $\eta^o$.
\end{itemize}	
\end{theorem}
\begin{proof}
	(1) It can be seen from the iterative process of algorithm AFT-PDASC that $\eta^j_{\A_j}=(\bar{X}_{\A_j}^{\top}\bar{X}_{\A_j})^{-1}\bar{X}_{\A_j}^{\top}\bar{Y}$. Recalling that $\bar{Y}=\bar{X}_{\A^*}\eta^*_{\A^*}+\gamma$, we have
	\begin{align*}
		\eta^j_{\A_j} -\eta_{\A_j}^*=\Big(\bar{X}_{\A_j}^{\top} \bar{X}_{\A_j}\Big)^{-1} \bar{X}_{\A_j}^{\top}\Big(\bar{X}_{\A^*}\eta^*_{\A^*}+\gamma-\bar{X}_{\A_j} \eta_{\A_j}^*\Big)=\Big(\bar{X}_{\A_j}^{\top} \bar{X}_{\A_j}\Big)^{-1} \bar{X}_{\A_j}^{\top}\Big(\bar{X}_{\B}\eta^*_{\B}+\gamma\Big).
	\end{align*}
	From the fact taht $\nu<\frac{1-2\alpha}{2K-1}$, $|\A_j|\le K$, and $0\le\alpha<\frac{1}{2}$, we can easily get that $(|\A_j|-1)\nu<1$.
	Then, using the inequalities in Lemma \ref{mipl}, we can further obtain that
	\begin{align*}
		\big\|\eta^j_{\A_j} -\eta_{\A_j}^*\big\|_{\infty}\leq \frac{1}{1-(|\A_j|-1)\nu}\Big(\|\bar{X}_{\A_j}^{\top} \bar{X}_{\B} \eta_{\B}^*\|_{\infty}+\|\bar{X}_{\A_j}^{\top} \gamma\|_{\infty}\Big)
		\le \frac{1}{1-(|\A_j|-1)\nu}\Big(\nu|\B|\cdot\|\eta^*_{\B}\|_{\infty}+\bar{\epsilon}\Big).
	\end{align*}
Additionally, based on the algebraic relationship between $\beta$ and $\eta$, we have
\begin{align*}
	\big\|\beta^j_{\A_j} -\beta_{\A_j}^*\big\|_{\infty}
	\leq\|D\|_{\infty} \big\|\eta^j_{\A_j} -\eta_{\A_j}^*\big\|_{\infty}\le \frac{\|D\|_{\infty}}{1-(|\A_j|-1)\nu}\Big(\nu|\B|\cdot\|D^{-1}\|_{\infty}\cdot\|\beta^*_{\B}\|_{\infty}+\bar{\epsilon}\Big),
\end{align*}
and
\begin{align*}
	\big\|\beta^j_{\A_j} -\beta_{\A_j}^*\big\|\leq \sqrt{|\A_j|}\big\|\beta^j_{\A_j} -\beta_{\A_j}^*\big\|_{\infty}\le \frac{\sqrt{K}\|D\|_{\infty}}{1-(|\A_j|-1)\nu}\Big(\nu|\B|\cdot\|D^{-1}\|_{\infty}\cdot\|\beta^*_{\B}\|_{\infty}+\bar{\epsilon}\Big).
\end{align*}

(2) From the fact $\nu<\frac{1-2\alpha}{2K-1}$, we know that $(2K-1)\nu+2\alpha<1$, which leads to $K \nu+\alpha<1-K \nu+\nu-\alpha$. Then, for any $s_1\in\big(\frac{1}{1-K \nu+\nu-\alpha},\frac{1}{K \nu+\alpha}\big)$, we have
$$
s_1>1+(K\nu-\nu+\alpha)s_1, \ \text{and} \ 1+(K\nu-\nu+\alpha)s_1>\frac{1}{1-K \nu+\nu-\alpha},
$$
i.e., $\frac{1}{1-K \nu+\nu-\alpha}<1+(K\nu-\nu+\alpha)s_1<s_1<\frac{1}{K \nu+\alpha}$.
Let $s_2:=1+(K\nu-\nu+\alpha)s_1$, then we have that $s_1, s_2 \in(1/(1-K \nu+\nu-\alpha), 1/(K \nu+\alpha))$ and that $s_1>s_2$.
Define the function $f(s_1):={s_2}/{s_1}=\big(1+(K\nu-\nu+\alpha)s_1\big)/{s_1}$, which has a monotonic decreasing property within the interval $\big(\frac{1}{1-K \nu+\nu-\alpha},\frac{1}{K \nu+\alpha}\big)$, then we have
$$
f\Big(\frac{1}{K \nu+\alpha}\Big)<f(s_1)<f\Big(\frac{1}{1-K \nu+\nu-\alpha}\Big),
$$ that is, $(2K-1)\nu+2\alpha<f(s_1)<1$.
This means that for any $\rho\in\big(((2K-1)\nu+2\alpha)^2,1\big)$, there exists a $s_1$ within the interval $\big(\frac{1}{1-K \nu+\nu-\alpha},\frac{1}{K \nu+\alpha}\big)$ such that ${s_2}/{s_1}=\sqrt{\rho}$.

In light of these analysis, for $k=0,1,2,\ldots,K_{\max}$, we now prove the following special monotonicity:
$$
\text{If} \ \T_{\lambda_k,s_1}\subseteq \A(\lambda_{k-1})\subseteq \A^*,\quad \text{then} \ \T_{\lambda_k,s_2}\subseteq \A(\lambda_{k})\subseteq \A^*.
$$
First, we consider the case of $k = 0$. For convenience, we define $\A(\lambda_{-1})=\A(\lambda_{0}):=\emptyset$. From Lemma \ref{mipl}, it holds that
$$
\|\eta^*\|_{\infty}=\|\eta_{\A^*}^*\|_{\infty}=\|(\bar{X}_{\A^*}^{\top}\bar{X}_{\A^*})^{-1}
(\bar{X}_{\A^*}^{\top}\bar{X}_{\A^*})\eta_{\A^*}^*\|_{\infty}\leq \frac{\|(\bar{X}_{\A^*}^{\top}\bar{X}_{\A^*})\eta_{\A^*}^*\|_{\infty}}{1-(K-1)\nu}.
$$
Furthermore, it also holds that
$$
\|\bar{X}^{\top}\bar{Y}\|_{\infty}\geq\|(\bar{X}_{\A^*}^{\top}\bar{X}_{\A^*})\eta_{\A^*}^*\|_{\infty}
-\|\bar{X}^{\top}\gamma\|_{\infty}\geq\big(1-(K-1)\nu\big)\|\eta^*\|_{\infty}-\bar{\epsilon}.
$$
Given $s_2>\frac{1}{1-K \nu+\nu-\alpha}$, we have $1-(K-1)\nu>{1}/{s_2}+\alpha$.
Additionally, using Assumption \ref{assum2}, we can obtain that $\|\bar{X}^{\top}\bar{Y}\|_{\infty}>\frac{1}{s_2}\|\eta^*\|_{\infty}$. This indicates that $\|\eta^*\|_{\infty}<s_2\|\bar{X}^{\top}\bar{Y}\|_{\infty}$.
Using the fact that $\lambda_0\geq\frac{1}{2}\|\bar{X}^{\top}\bar{Y}\|_{\infty}^2$ and the definition of $\T_{\lambda,s}$, we can get $\T_{\lambda_0,s_1}=\emptyset$ and $\T_{\lambda_0,s_2}=\emptyset$.
Therefore, it holds that $\T_{\lambda_0,s_1}\subseteq \A(\lambda_{-1})\subseteq \A^*$ and that $\T_{\lambda_0,s_2}\subseteq \A(\lambda_{0})\subseteq \A^*$.

Assume that the monotonicity \eqref{col22}  holds when $k = l$, that is, ``if $ \T_{\lambda_l,s_1}\subseteq \A(\lambda_{l-1})\subseteq \A^*$, then $\T_{\lambda_l,s_2}\subseteq \A(\lambda_{l})\subseteq \A^*$". Then, we will show that this assertion also holds when $k = l + 1$.
From the relations $\lambda_{l+1}=\rho\lambda_l$ and $\rho=s_2^2/s_1^2$, we know that $\T_{\lambda_{l+1},s_1}=\T_{\lambda_l,s_2}$. Based on the result for $k = l$, we have $\T_{\lambda_{l+1},s_1}\subseteq \A(\lambda_{l})\subseteq \A^*$. Therefore, it remains to prove $\T_{\lambda_{l+1},s_2}\subseteq \A(\lambda_{l+1})\subseteq \A^*$.
Note that $\A(\lambda_{l})$ and $\A(\lambda_{l+1})$ are essentially the initial guess and final output of the active set in the $\lambda_{l+1}$-problem, respectively. Therefore,  we only need to demonstrate that the following assertion holds during the inner iterations of the $\lambda_{l+1}$-problem, i.e.,
      \begin{align}\label{lps1}
      \text{If} \ \T_{\lambda_{l+1},s_1}\subseteq \A_j\subseteq \A^*, \quad \text{then} \ \T_{\lambda_{l+1},s_2}\subseteq \A_{j+1}\subseteq \A^*.
      \end{align}

From the relation $\eta^j_{\A_j}=\eta^j_{\A_j}-\eta^*_{\A_j}+\eta^*_{\A_j}$ for any $i\in \A_j$, we get
$$|\eta^j_i|\geq |\eta^*_i|-\|\eta^{j}_{\A_j} -\eta_{\A_j}^* \|_{\infty}\geq |\eta^*_i|-\frac{\nu|\B|\cdot\|\eta_{\B}^*\|_{\infty}+\bar{\epsilon}}{1-(|\A_j|-1) \nu}.$$
Additionally, from the update formula of the dual variable $d$, it can be seen that
\begin{align*}
	d_i^j&=\bar{X}^{\top}_i\Big(\bar{Y}-\bar{X}_{\A_j}\eta^j_{\A_j}\Big)=
	\bar{X}^{\top}_i\Big(\bar{X}_{\A^*}\eta^*_{\A^*}+\gamma-\bar{X}_{\A_j}\eta^j_{\A_j}\Big) \\
	&=\bar{X}^{\top}_i\Big(\bar{X}_{\A^*}\eta^*_{\A^*}-\bar{X}_{\A_j}\eta^*_{\A_j}+\bar{X}_{\A_j}\eta^*_{\A_j}
	+\gamma-\bar{X}_{\A_j}\eta^j_{\A_j}\Big)\\
	&=\bar{X}^{\top}_i\Big(\bar{X}_{\B}\eta^*_{\B}+\gamma-\bar{X}_{\A_j}(\eta^j_{\A_j}-\eta^*_{\A_j})\Big).
\end{align*}
Then, for any $i\in \B$, using the result in (1), Lemma \ref{mipl}, and $\bar{X}^{\top}_i\bar{X}_i=1$, we know that
\begin{align}\label{d1}
	|d^j_i|&=\Big|\bar{X}^{\top}_i\bar{X}_i\eta^*_i+\bar{X}^{\top}_i\Big(\bar{X}_{\B\backslash \{i\}}\eta^*_{\B \backslash \{i\}}+\gamma-\bar{X}_{\A_j}(\eta^j_{\A_j}-\eta^*_{\A_j})\Big)\Big|\notag\\
	&\geq |\eta^*_i|-|\bar{X}^{\top}_i\bar{X}_{\B\backslash \{i\}}\eta^*_{\B \backslash \{i\}}|-|\bar{X}^{\top}_i\gamma|-|\bar{X}^{\top}_i\bar{X}_{\A_j}(\eta^j_{\A_j}-\eta^*_{\A_j})|\notag\\
	&\geq |\eta^*_i|-(|\B|-1)\nu\|\eta^*_{\B}\|_{\infty}-\bar{\epsilon}-|\A_j|\nu\|\eta^j_{\A_j}-\eta^*_{\A_j}\|_{\infty}\notag\\
	&\geq |\eta^*_i|+\nu\|\eta^*_{\B}\|_{\infty}-|\B|\nu\Big(1+\frac{|\A_j|\nu}{1-(|\A_j|-1)\nu}\Big)\|\eta^*_{\B}\|_{\infty}
	-\bar{\epsilon}\Big(1+\frac{|\A_j|\nu}{1-(|\A_j|-1)\nu}\Big).
\end{align}
Similarly, for any $i\in\I^*$, we have
\begin{align}\label{d2}
	|d^j_i|&\leq |\bar{X}^{\top}_i\bar{X}_{\B}\eta^*_{\B}|+|\bar{X}^{\top}_i\gamma|+|\bar{X}^{\top}_i\bar{X}_{\A_j}
	(\eta^j_{\A_j}-\eta^*_{\A_j})|\notag\\
	&\leq|\B|\nu\|\eta^*_{\B}\|_{\infty}+\bar{\epsilon}+|\A_j|\nu\|\eta^j_{\A_j}-\eta^*_{\A_j}\|_{\infty}\notag\\
	&\leq |\B|\nu\Big(1+\frac{|\A_j|\nu}{1-(|\A_j|-1)\nu}\Big)\|\eta^*_{\B}\|_{\infty}
	+\bar{\epsilon}\Big(1+\frac{|\A_j|\nu}{1-(|\A_j|-1)\nu}\Big).
\end{align}
From $|\A_j|=K-|\B|$, $\frac{|\B|\nu+\alpha}{1-K\nu+\nu+|B|\nu} \leq\frac{K \nu+\alpha}{1+\nu}$, and $\bar{\epsilon} \leq \alpha \min _{i \in \A^*}|\eta_i^*| \leq \alpha\|\eta_{\B}^*\|_{\infty}$, we get
\begin{align*}
	&|\B|\nu\Big(1+\frac{|\A_j|\nu}{1-(|\A_j|-1)\nu}\Big)\|\eta^*_{\B}\|_{\infty}
	+\bar{\epsilon}\Big(1+\frac{|\A_j|\nu}{1-(|\A_j|-1)\nu}\Big)\\
	\leq &\frac{|\B| \nu+\alpha}{1-K \nu+\nu+|\B| \nu}(1+\nu)\|\eta_{\B}^*\|_{\infty}\leq (K \nu+\alpha)\|\eta_{\B}^*\|_{\infty}.
\end{align*}
Therefore, it can be further concluded from (\ref{d1}) and (\ref{d2}) that
\begin{align*}
	|d^j_i|\geq|\eta^*_i|-(K\nu-\nu+\alpha)\|\eta_{\B}^*\|_{\infty}, \ \forall i\in \B, \quad\text{and}\quad
	|d^j_i|\leq(K \nu+\alpha)\|\eta_{\B}^*\|_{\infty}, \ \forall i\in\I^*.
\end{align*}
Using $\T_{\lambda_{l+1},s_1}\subseteq \A_j$, we get $\|\eta_{\B}^*\|_{\infty}<s_1\sqrt{2\lambda_{l+1}}$.
Next, we will prove $\T_{\lambda_{l+1},s_2}\subseteq \A_{j+1}$.
For any $i\in \I_j \cap \T_{\lambda_{l+1},s_2}$, it holds that
$$|d_i^j|>s_2 \sqrt{2 \lambda_{l+1}}-(K\nu-\nu+\alpha)s_1 \sqrt{2 \lambda_{l+1}}=[s_2-(K\nu-\nu+\alpha)s_1]\sqrt{2 \lambda_{l+1}} = \sqrt{2 \lambda_{l+1}},$$
which means that $i \in \A_{j+1}$. For any $i\in \A_j \cap \T_{\lambda_{l+1},s_2}$, from the relation $\frac{|\B| \nu+\alpha}{1-(|\A_j|-1) \nu}\leq\frac{|\B| \nu+\alpha+(|\A_j|-1) \nu}{1-(|\A_j|-1) \nu+(|\A_j|-1) \nu}=K \nu-\nu+\alpha$, we can deduce that
$$|\eta_i^j| \geq|\eta_i^*|-\frac{|\B| \nu+\alpha}{1-(|\A_j|-1) \nu}\|\eta_{\B}^*\|_{\infty}>s_2 \sqrt{2 \lambda_{l+1}}-(K \nu-\nu+\alpha) s_1 \sqrt{2 \lambda_{l+1}}=\sqrt{2 \lambda_{l+1}},$$
which means $i \in \A_{j+1}$.
Therefore, $\T_{\lambda_{l+1},s_2}\subseteq \A_{j+1}$.
Then, we get conclusion $\A_{j+1}\subseteq\A^*$ by proving $\I^*\subseteq\I_{j+1}$. For all $i\in \I^*$, we have $|d^j_i|<s_1(K \nu+\alpha)\sqrt{2 \lambda_{l+1}}<\sqrt{2 \lambda_{l+1}}$, which means $i \in \I_{j+1}$.
By this point, we have proven the assertion (\ref{lps1}).

 Furthermore, we have proven that, for $k=0,1,2,\ldots,K_{\max}$, the active sets ${\A(\lambda_k)}$ have the following special monotonicity:
$$
\text{If} \ \T_{\lambda_k,s_1}\subseteq \A(\lambda_{k-1})\subseteq \A^*,\quad \text{then} \ \T_{\lambda_k,s_2}\subseteq \A(\lambda_{k})\subseteq \A^*.
$$
From the above conclusion, it is clear that the active set ${\A(\lambda_k)}$ is always contained within $\A^*$. Additionally, for (\ref{2model2}), we define $\eta_{\lambda_1}$ and $\eta_{\lambda_2}$ as the optimal solutions corresponding to $\lambda_1$ and $\lambda_2$, respectively. Here, we assume that $\lambda_1>\lambda_2$. Then, we have
\begin{align*}
	&\frac{1}{2} \|\bar{Y}-\bar{X} \eta_{\lambda_1}\|^2 +\lambda_1\|\eta_{\lambda_1}\|_0\leq\frac{1}{2} \|\bar{Y}-\bar{X} \eta_{\lambda_2}\|^2 +\lambda_1\|\eta_{\lambda_2}\|_0,\\
	&\frac{1}{2} \|\bar{Y}-\bar{X} \eta_{\lambda_2}\|^2 +\lambda_2\|\eta_{\lambda_2}\|_0\leq\frac{1}{2} \|\bar{Y}-\bar{X} \eta_{\lambda_1}\|^2 +\lambda_2\|\eta_{\lambda_1}\|_0.
\end{align*}
Adding both sides of the two inequalities above, we can obtain $\lambda_1\|\eta_{\lambda_1}\|_0+\lambda_2\|\eta_{\lambda_2}\|_0\leq\lambda_1\|\eta_{\lambda_2}\|_0+\lambda_2\|\eta_{\lambda_1}\|_0$, i.e., $(\lambda_1-\lambda_2)(\|\eta_{\lambda_1}\|_0-\|\eta_{\lambda_2}\|_0)\leq0$. Therefore, we get $\|\eta_{\lambda_1}\|_0\leq\|\eta_{\lambda_2}\|_0$. This indicates that as $k$ increases, $\lambda_k$ continuously decreases, and ${\A(\lambda_k)}$ will eventually either be $\A^*$ itself or a proper subset of $\A^*$.
In a word, both cases  lead to the termination criterion in the $12$-th line of AFT-PDASC is satisfied, so that the algorithm is terminated in finite steps.

(3) From the proof process we know that, when AFT-PDASC terminates, the estimated active set $\hat{\A}$ is either $\A^*$ itself or a proper subset of $\A^*$. Next, we will proceed by contradiction to prove that the final active set must be exactly $\A^*$ itself but cannot be a proper subset.
We assume that $\hat{\A}\subsetneqq \A^*$, and then $\B:=\A^*\backslash \hat{\A}$ is non-empty.
Let $\hat{\I}:=\S \setminus \hat{\A}$. Noting that $\hat{\eta}$ represents the estimated regression coefficient. Then, we have
\begin{align*}
	L(\hat{\eta})=\frac{1}{2} \|\bar{Y}-\bar{X} \hat{\eta}\|^2 +\lambda\|\hat{\eta}\|_0
	%&=\frac{1}{2} \|\bar{X}_{\A^*}\eta^*_{\A^*}+\gamma-\bar{X}_{\hat{\A}} \hat{\eta}_{\hat{\A}}\|^2 +\lambda|\hat{\A}|\\
	=\frac{1}{2} \|\bar{X}_{\B}\eta^*_{\B}+\gamma-\bar{X}_{\hat{\A}} (\hat{\eta}_{\hat{\A}}-\eta^*_{\hat{\A}})\|^2 +\lambda|\hat{\A}|.
\end{align*}
Let $i_{\A^*}\in \{i\in\hat{\I} : |\eta^*_i|=\|\eta^*_{\B}\|_{\infty}\}$, then it is easy to see that $i_{\A^*}\in \B$ and $|\eta^*_{i_{\A^*}}|=\|\eta^*_{\B}\|_{\infty}$. Furthermore, we have
\begin{align*}
	L(\hat{\eta})&=\frac{1}{2} \|\bar{X}_{i_{\A^*}}\eta^*_{i_{\A^*}}+\bar{X}_{\B\backslash \{i_{\A^*}\}}\eta^*_{\B\backslash\{i_{\A^*}\}}+\gamma-\bar{X}_{\hat{\A}} (\hat{\eta}_{\hat{\A}}-\eta^*_{\hat{\A}})\|^2 +\lambda|\hat{\A}|\\
	&\geq \frac{1}{2}|\eta^*_{i_{\A^*}}|^2-|\eta^*_{i_{\A^*}}|\Big( |\langle\bar{X}_{i_{\A^*}},\bar{X}_{\B\backslash \{i_{\A^*}\}}\eta^*_{\B\backslash\{i_{\A^*}\}}\rangle|+|\langle\bar{X}_{i_{\A^*}},\gamma\rangle| +|\langle\bar{X}_{i_{\A^*}},\bar{X}_{\hat{\A}} (\hat{\eta}_{\hat{\A}}-\eta^*_{\hat{\A}})\rangle|\Big)+\lambda|\hat{\A}|\\
	&\geq\frac{1}{2}|\eta^*_{i_{\A^*}}|^2-|\eta^*_{i_{\A^*}}|\Big((|\B|-1)\nu|\eta^*_{i_{\A^*}}|+\bar{\epsilon} +|\hat{\A}|\nu\|\hat{\eta}_{\hat{\A}}-\eta^*_{\hat{\A}}\|_{\infty} \Big)+\lambda|\hat{\A}|\\
	&\geq \frac{1}{2}|\eta^*_{i_{\A^*}}|^2-|\eta^*_{i_{\A^*}}|\Big((|\B|-1)\nu|\eta^*_{i_{\A^*}}|+\bar{\epsilon} +\frac{|\hat{\A}|\nu}{1-(|\hat{\A}|-1)\nu}(|\B|\nu|\eta^*_{i_{\A^*}}|+\bar{\epsilon}) \Big)+\lambda|\hat{\A}|.
\end{align*}
In addition, using Assumption \ref{assum2}, we have $\bar{\epsilon}\le\alpha \min _{i \in \A^*}\{|\eta_i^*|\} \leq \alpha|\eta_{i_{\A^*}}^*|$. Then, we can deduce that
\begin{align*}
	L(\hat{\eta})&\geq\frac{1}{2}|\eta^*_{i_{\A^*}}|^2-|\eta^*_{i_{\A^*}}|\Big((|\B|-1)\nu|\eta^*_{i_{\A^*}}|+
	\alpha|\eta_{i_{\A^*}}^*| +\frac{|\hat{\A}|\nu}{1-(|\hat{\A}|-1)\nu}(|\B|\nu|\eta^*_{i_{\A^*}}|+\alpha|\eta_{i_{\A^*}}^*|) \Big)+\lambda|\hat{\A}|\\
	&=|\eta^*_{i_{\A^*}}|^2 \Big( \frac{1}{2}-(|\B|-1)\nu-\alpha-\frac{|\hat{\A}|\nu(|\B|\nu+\alpha)}{1-(|\hat{\A}|-1)\nu} \Big)+\lambda|\hat{\A}|\\
	&=|\eta^*_{i_{\A^*}}|^2 \Big( \frac{1}{2}-(K-1)\nu-\alpha \Big)+|\eta^*_{i_{\A^*}}|^2|\hat{\A}|\nu\Big(1-\frac{|\B|\nu+\alpha}{1-(|\hat{\A}|-1)\nu}\Big)+\lambda|\hat{\A}|.
\end{align*}
Here, $1-\frac{|\B|\nu+\alpha}{1-(|\hat{\A}|-1)\nu}=\frac{1-\big((|\hat{\A}|+|\B|-1)\nu+\alpha\big)}{1-(|\hat{\A}|-1)\nu}$.
It has been analyzed earlier that $(|\hat{\A}|-1)\nu<1$, and $(|\hat{\A}|+|\B|-1)\nu+\alpha<(2K-1)\nu+2\alpha<1$, so we get $1-\frac{|\B|\nu+\alpha}{1-(|\hat{\A}|-1)\nu}>0$. Therefore, we can conclude that $L(\hat{\eta})\geq|\eta^*_{i_{\A^*}}|^2 \Big( \frac{1}{2}-(K-1)\nu-\alpha \Big)$.

Next, we will first explain the reasonableness of the setting for the value range of $\lambda$. From $0\leq \alpha<\frac{1}{2}$, we have $\alpha^2<\alpha$. With the relation of $\alpha\le\frac{1-2(K-1)\nu}{K+3}$, we can easily obtain $(K+1) \alpha^2+2 \alpha<(K+3) \alpha \leq 1-2(K-1) \nu$. Furthermore, it has $K \alpha^2<1-2(K-1) \nu-2 \alpha-\alpha^2$.
It is clear from the definition of $\xi$ and Assumption \ref{assum2} that $\xi>\frac{\alpha^2}{2} \min _{i \in \A^*}|\eta_i^*|^2\geq \bar{\epsilon}^2/2$. Therefore, the setting of the value range of $\lambda$ is reasonable.
Based on the value of $\lambda$, we have
\begin{align*}
	L(\hat{\eta})-L(\eta^o)&\geq |\eta^*_{i_{\A^*}}|^2 \Big( \frac{1}{2}-(K-1)\nu-\alpha \Big)-\frac{1}{2}\bar{\epsilon}^2-\lambda K\\
	&\geq \Big( \frac{1}{2}-(K-1)\nu-\alpha-\frac{\alpha^2}{2} \Big)\min _{i \in \A^*}\{|\eta_i^*|^2\}-\lambda K\\
	&=\xi K-\lambda K>0.
\end{align*}
In other words, it gets $\frac{1}{2}\|\bar{Y}-\bar{X}\hat{\eta}\|^2+\lambda|\hat{\A}|>\frac{1}{2}\bar{\epsilon}^2+\lambda K$.
Since $|\hat{\A}|<K$, then it yields $\|\bar{Y}-\bar{X}\hat{\eta}\|>\bar{\epsilon}$. This contradicts the stopping condition in $12$-th line of  AFT-PDASC, so the previous assumption is incorrect. In other words, $\hat{\A}$ cannot be a proper subset of $\A^*$, it must be $\A^*$ itself. This also indicates that AFT-PDASC terminates at the oracle solution $\eta^o$.
\end{proof}

At the end of this section, using the Assumption \ref{assum} (2) that the covariate matrix $\bar{X}$ satisfies the RIP condition, we also present the error analysis, the monotonicity of the active set, and the finite-step termination of AFT-PDAS at the oracle solution $\eta^o$. Since the proof process is similar to that of Theorem \ref{th1}, we omit it here.

\begin{theorem}\label{th2}
	Suppose that Assumption \ref{assum2} holds. Let $\delta:=\delta_{K+1}<\frac{1-2\alpha}{2\sqrt{K}+1}$, then for any $\rho \in\left((\frac{2\delta\sqrt{K}+2\alpha}{1-\delta})^2, 1\right)$, it holds that:
	\begin{itemize}
		\item[(1)] In the $\lambda_k$ subproblem, for $j=1,2,\ldots,J_{\max}$, we set $\B:=\A^* \backslash \A_j $ and $\I_j=\S \backslash \A_j$. If $|\A_j|\le K$, then we have
		$$
		\begin{aligned}
			\left\|\eta^j_{\A_j} -\eta_{\A_j}^* \right\| & \leq \frac{\delta_{|\A_j|+|\B|}}{1-\delta_{|\A_j|}}\left\|\eta_{\B}^*\right\|+\frac{1}{\sqrt{1-\delta_{|\A_j|}}} \bar{\epsilon}, \\
			\left\|\beta^j_{\A_j} -\beta_{\A_j}^* \right\| & \leq \frac{\|D\|\delta_{|\A_j|+|\B|}}{1-\delta_{|\A_j|}}\left\|D^{-1}\beta_{\B}^*\right\|+\frac{\|D\|}{\sqrt{1-\delta_{|\A_j|}}} \bar{\epsilon}.
		\end{aligned}
		$$
		\item[(2)] For $k=0,1,2,\ldots,K_{\max}$, there exist $s_1, s_2 \in(\frac{ 1-\delta}{1-\delta- \delta\sqrt K- \alpha}, \frac{ 1-\delta}{\delta\sqrt K+\alpha})$ with $s_1>s_2$, such that the active sets ${\A(\lambda_k)}$ have the following special monotonicity:
		$$ \text{If} \ \T_{\lambda_k,s_1}\subseteq \A(\lambda_{k-1})\subseteq \A^*, \quad\text{then} \ \T_{\lambda_k,s_2}\subseteq \A(\lambda_{k})\subseteq \A^*.$$
		Furthermore, AFT-PDASC   terminates in a finite number of steps.
		\item[(3)] Let $\alpha \leq(1-2\delta-\delta^2) /4$ and $\xi=( \frac{1-\delta}{2}-\frac{\delta^2}{1-\delta}-\frac{\alpha}{\sqrt{1-\delta}}-\frac{1}{2} \alpha^2) \min _{i \in \A^*}\{|\eta_i^*|^2\}$. Then, for any $\lambda \in(\frac{\bar{\epsilon}^2}{2}, \xi)$, AFT-PDASC terminates at the oracle solution $\eta^o$.
	\end{itemize}	
\end{theorem}

\section{Numerical experiments}\label{num}
\setcounter{equation}{0}
%%%%%%%%%%%%%%%%%%%%%%%%%%%%%%%%%%%%%%%%%%%%
In this section, we demonstrate the progressiveness of the estimation method in \eqref{2model1} and highlight the effectiveness of the AFT-PDASC algorithm for handling censored data from both numerical simulations and real-world applications.
We also conduct performance comparisons with several state-of-the-art estimation methods, focusing on the impact of model parameters and the accuracy of the estimations.
The algorithms selected for comparison include the support detection and root finding algorithm (SDAR) \cite{CFHJZ2022}, Orthogonal Matching Pursuit (OMP) \cite{PRK1993}, Greedy Gradient Pursuit (GreedyGP) \cite{BD2008}, Accelerated Iterative Hard Thresholding (AIHT) \cite{B2012}, and Hard Thresholding Pursuit (HTP) \cite{F2011}.
All the experiments are performed with Microsoft Windows 11 and MATLAB R2022a, and run on a PC with an Intel(R) Xeon(R) W-2295 CPU at 3GHz and 128 GB of memory.

We now outline the data generation process used in the simulation study.
We first generate a $n\times p$ random Gaussian matrix $\check{X}$ whose entries
are i.i.d. $\N (0, 1)$. Then $X$ is generated with $X_1:=\check{X}_1$, $X_p:=\check{X}_p$ and $X_i:=\check{X}_i+\kappa(\check{X}_{i+1}+\check{X}_{i-1})$ for $i=2, \ldots, p-1$. Here, $\kappa$ measures the strength of the correlation between the covariates.
To generate the true regression coefficient $\beta^*$, we first randomly select a subset of $\S$ to form the true active set $\A^*$. Let $R:=m_2/m_1$, where $m_2=\max\{|\beta^*_i|: i\in \A^*\}$ and $m_1=\min\{|\beta^*_i|: i\in \A^*\}$.
Subsequently, the $K$ nonzero coefficients in $\beta^*$ are distributed uniformly within the interval $(m_1,m_2)$.
Then we set $\ln(T_i):=X_i^{\top}\beta^*+\epsilon_i$ for $i=1,\ldots,n$, where $\epsilon_i$ is generated independently from $\N(0,\sigma^2)$.
The censoring time $C_i$ follows a uniform distribution $U(0, \zeta)$, where $\zeta$ controls the censoring rate.
Then for $i=1,\ldots,n$, the response variable is generated by $Y_i:=\text{min}\{\text{ln}(T_i),\text{ln}(C_i)\}$.
In  AFT-PDASC, we use a grid search method to select the appropriate regularization parameter. Specifically, we choose $\lambda_0=\frac{1}{2}\|\bar{X}^{\top}\bar{Y}\|_{\infty}^2$ and set $\lambda_{\min}:=10^{-15}\lambda_0$, and then divide the interval $[\lambda_{\min},\lambda_0]$ into $N$ equally spaced subintervals. It is easy to observe that as $N$ increases, the decay factor $\rho$ also becomes larger. The values of the other parameters will be given as they occur.

\subsection{Performance evaluation on a simple simulated data}

In this part, we analyze the numerical performance of algorithm AFT-PDASC  using $100$ independent trials. Here, we choose $n=500$, $p=1000$, $K=10$, $\kappa=0.3$, $N=100$, and $J_{\max}=2$. Let $c.r$ be the censoring rate, and in this case, we consider $c.r=0.3$. At the same time, we fix the $10$ non-zero elements of the true regression coefficients $\beta^*$ to be $\beta^*_{34}=5$, $\beta^*_{166}=-1$, $\beta^*_{278}=2$, $\beta^*_{354}=-3$, $\beta^*_{409}=4$, $\beta^*_{520}=-5$, $\beta^*_{666}=1$, $\beta^*_{708}=-4$, $\beta^*_{821}=3$, $\beta^*_{942}=-2$.
The boxplot in Figure \ref{fig1} illustrates the estimation performance of AFT-PDASC in this test.

\begin{figure}[h]
	\centering\vspace{-.5cm}
	\includegraphics[width=0.6\textwidth]{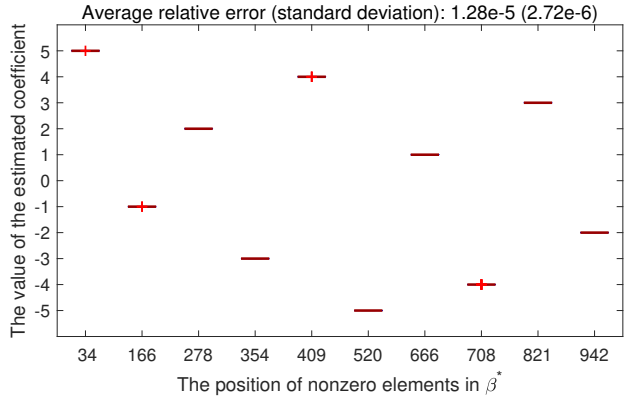}
	\vspace{-.8cm}
	\caption{\scriptsize The boxplot of AFT-PDASC  under $100$ independent trials}
    \label{fig1}
\end{figure}

From Figure \ref{fig1}, it is evident that the AFT-PDASC algorithm not only accurately locates each non-zero element but also consistently estimates their values correctly in nearly every trial. Additionally, the average relative error and standard deviation between the true regression coefficient and the estimated coefficient are very small, further highlighting the excellent estimation performance of the method in \eqref{2model1} and the effectiveness of the AFT-PDASC algorithm.

\subsection{Performance comparisons with other state-of-the-art algorithms}

In this part, we conduct the performance comparisons of the AFT-PDASC algorithm with SDAR, OMP, GreedyGP, AIHT, and HTP, focusing on the impact of model parameters and estimation performance.
To demonstrate the numerical stability of each algorithm, we perform a simulation analysis based on the results of $10$ independent repetitions.
Let $\widehat{\A}^{(i)}$ be the active set obtained in the $i$-th experiment.
In addition to recording the computing time (Time) and average relative error (ReErr), we introduce an indicator defined as ``$\text{Probability}=\frac{1}{10} \sum_{i=1}^{10} \mathbf{1}_{\{\widehat{\A}^{(i)}=\A^*\}}$", which measures the probability of accurately recovering the true active set.

\subsubsection{Parameters' values influence evaluation}
In this test, we examine the impact of certain parameter values, namely $\{n, p, K, \kappa\}$, on the active set recovery performance of each algorithm. The specific values of all parameters for each experimental setting are as follows:
\begin{itemize}
	\item[(i)] $n=\{200:50:500\}$, $p=1000$, $K=15$, $R=10^3$, $\kappa=0.3$, $c.r=0.3$, $\sigma=1e-3$, $N=150$, $J_{\max}=6$.
	\item[(ii)] $n=500$, $p=\{1200:100:2000\}$, $K=15$, $R=10^3$, $\kappa=0.3$, $c.r=0.3$, $\sigma=1e-3$, $N=150$, $J_{\max}=6$.
	\item[(iii)] $n=500$, $p=1000$, $K=\{10:10:200\}$, $R=10^3$, $\kappa=0.3$, $c.r=0.3$, $\sigma=1e-3$, $N=150$, $J_{\max}=6$.
	\item[(iv)] $n=500$, $p=1000$, $K=15$, $\kappa=\{0.1:0.1:0.8\}$, $R=10^3$, $c.r=0.3$, $\sigma=1e-3$, $N=150$, $J_{\max}=6$.
\end{itemize}
Figure \ref{fig2} shows the results of ``Probability" for the six algorithms under different values of the parameters $\{n, p, K, \kappa\}$.
It is evident from the figure that the regression probability of the AFT-PDASC algorithm is consistently higher than that of the other algorithms, with the value of "Probability" approaching $1$ in most cases.
This phenomenon not only demonstrates that the AFT-PDASC algorithm is more stable compared to other methods, but also indicates that it can achieve accurate recovery across a wide range of parameter settings.

\begin{figure}
\centering
\includegraphics[width=0.4\textwidth]{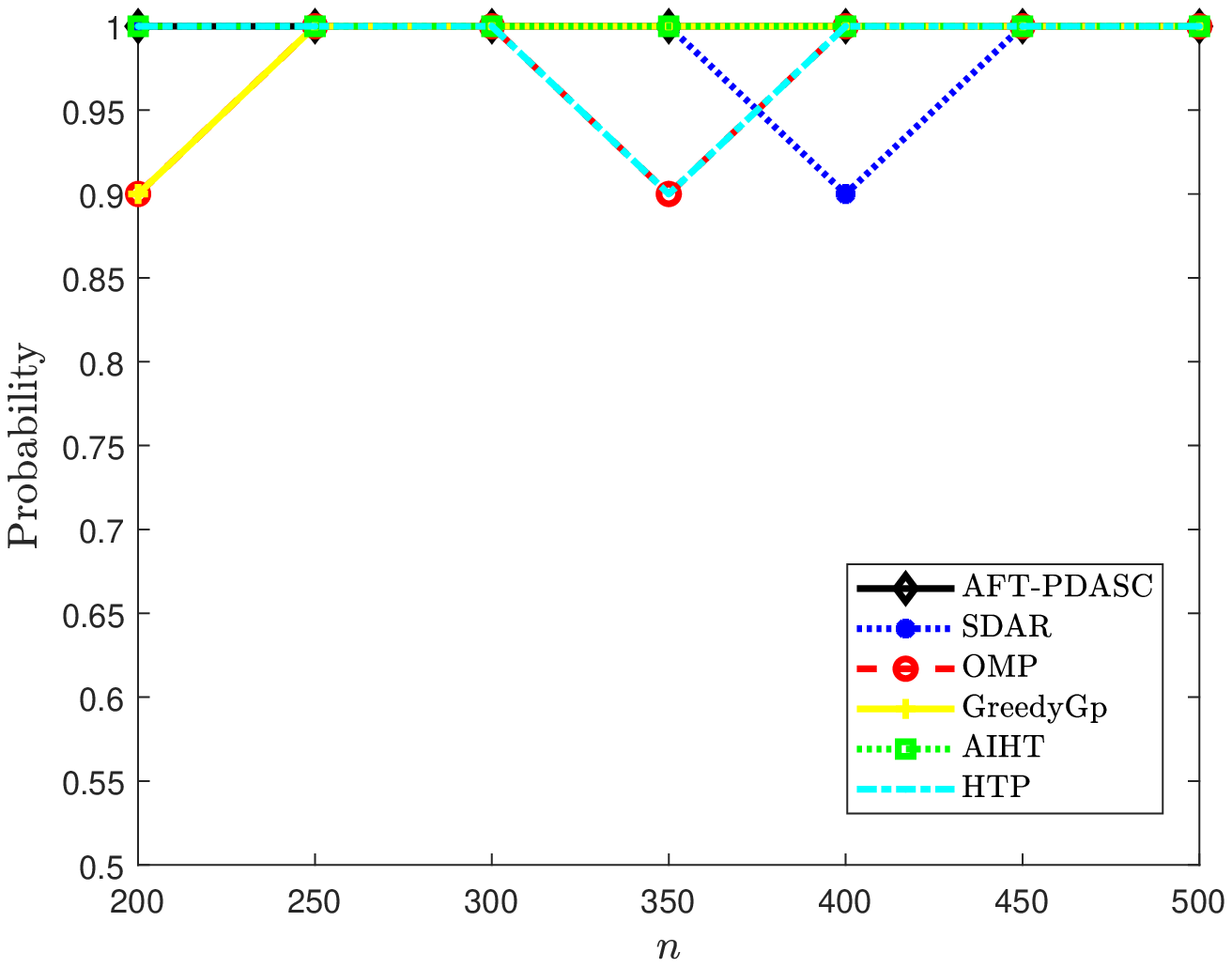}\hspace{-.4cm}
\includegraphics[width=0.4\textwidth]{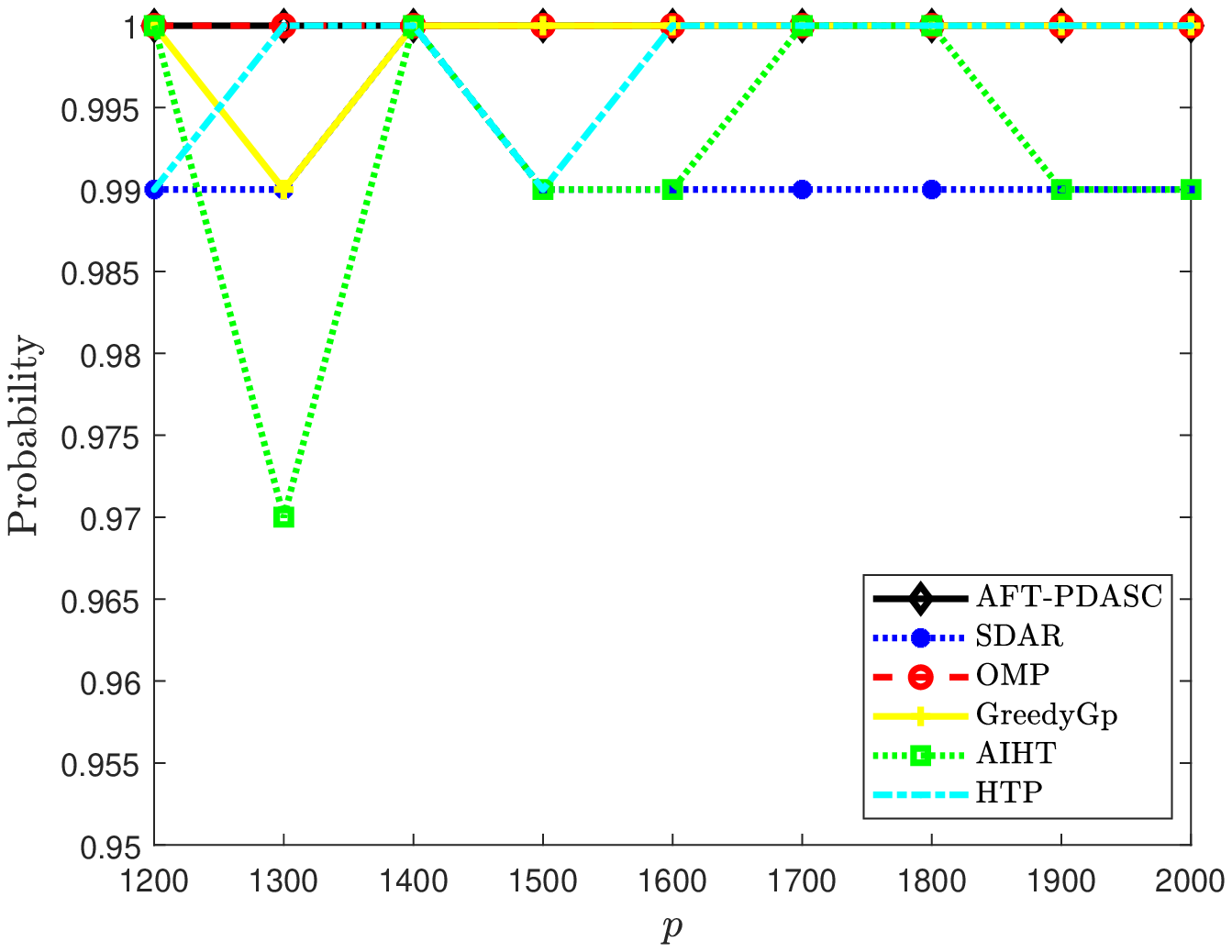}\\
\includegraphics[width=0.4\textwidth]{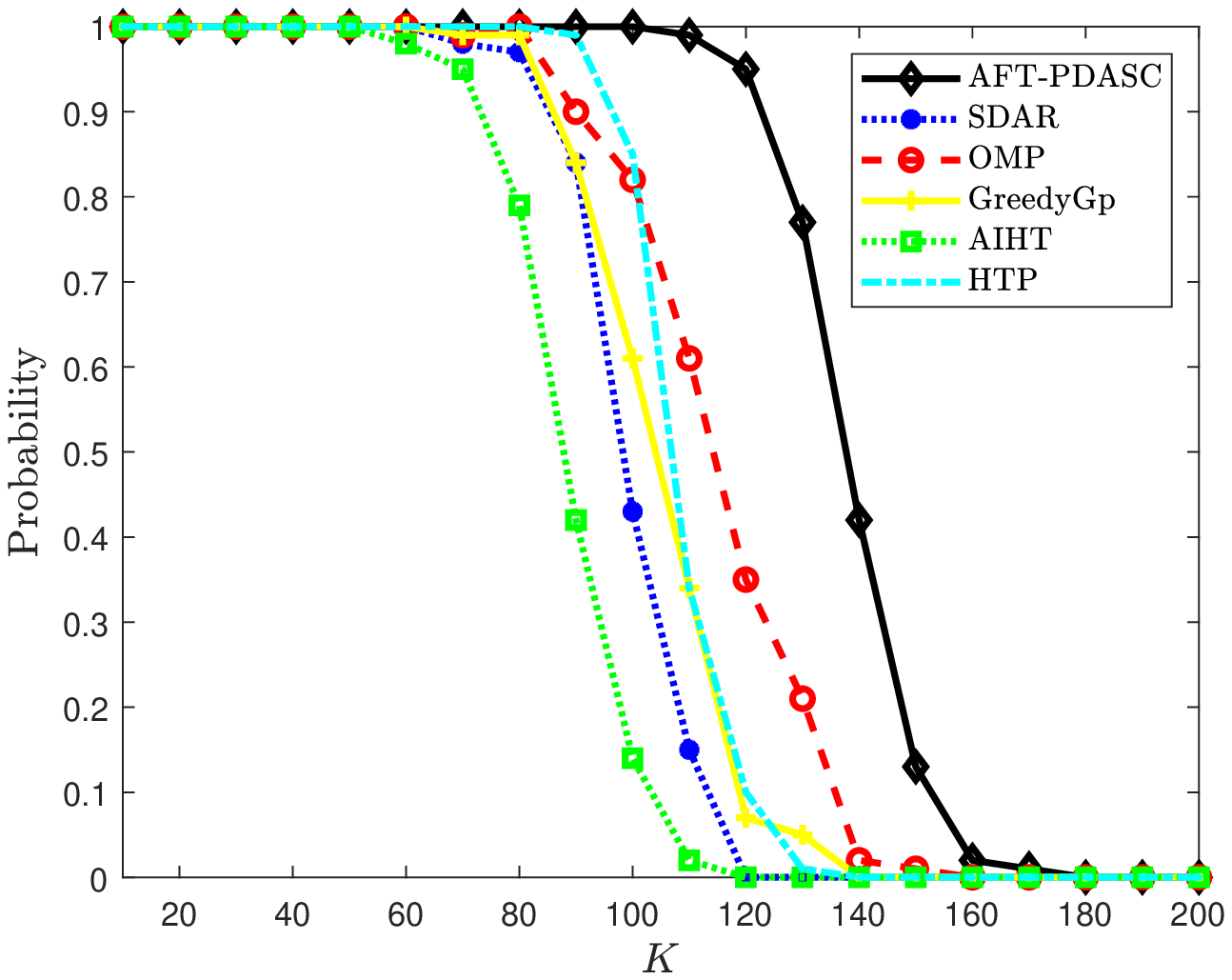}\hspace{-.4cm}
\includegraphics[width=0.4\textwidth]{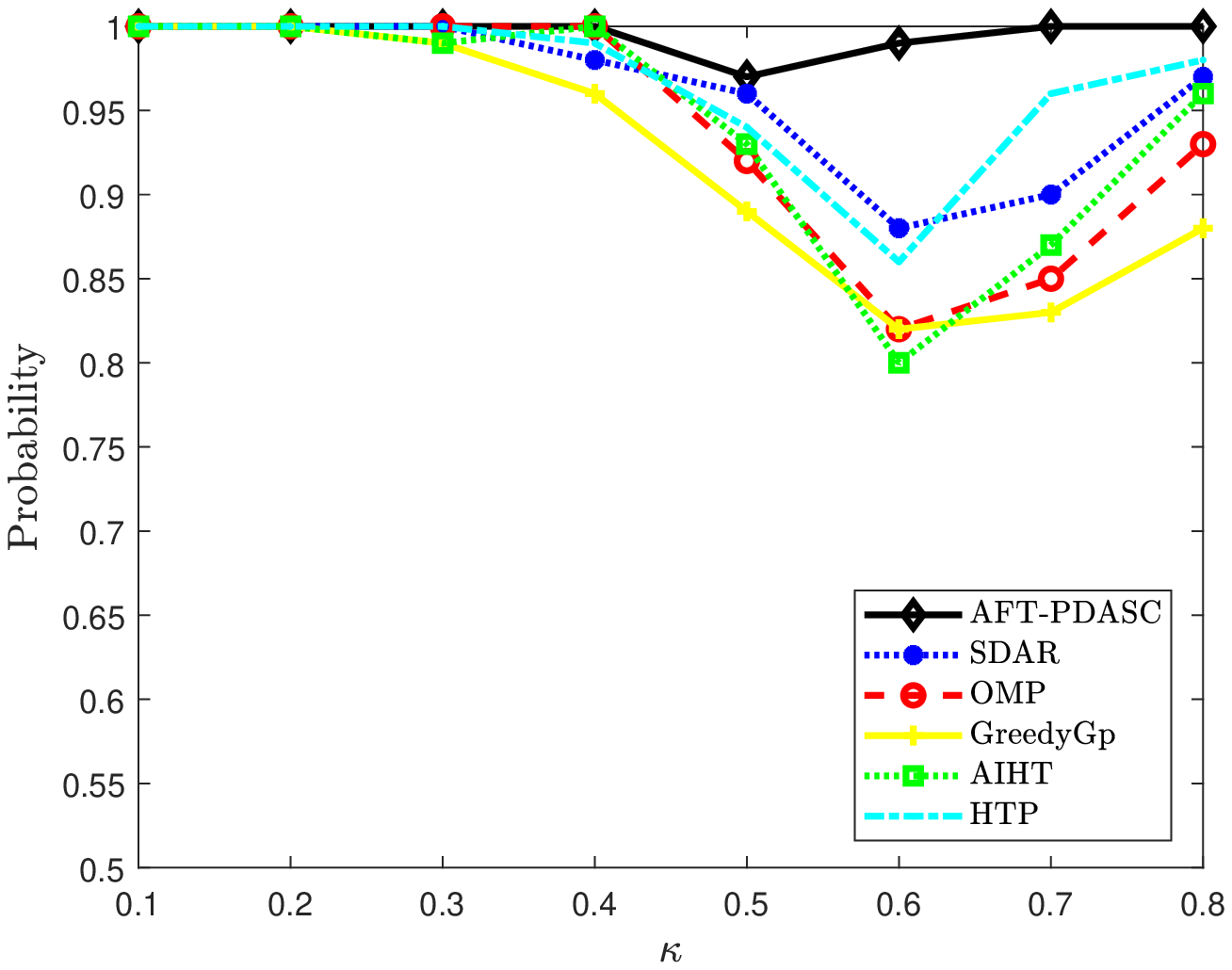}
\caption{\scriptsize The ``Probability" values of each algorithm under different parameters' settings}
\label{fig2}
\end{figure}

%%%%%%%%%%%%%%%%%%%%%%%%%%%%%%%%%%%%%%%%%%%%%%5
\subsubsection{Performance comparisons on simulated data}
%%%%%%%%%%%%%%%%%%%%%%%%%%%%%%%%%%%%%%%%%%%%%%%%
To further highlight the accuracy and efficiency of the AFT-PDASC algorithm, we compare it with the other algorithms in terms of three key metrics: computing time (Time), relative error (ReErr), and recovery probability (Probability).
In this test, we consider three different levels of correlation for the simulation matrix by setting $\kappa=0.1:0.3:0.7$. Additionally, we consider three different dimensions by setting $p = 1000$, $2000$ and $3000$. Other parameters's values are fixed as $n=500$, $K=20$, $R=10^3$, $c.r=0.3$, $\sigma=1e-3$, $N=100$, and $J_{\max}=2$. The results of the different algorithms in terms of Time, ReErr and Probability are listed in Table \ref{tab1}.

\renewcommand{\arraystretch}{0.7}
\begin{table}
	\setlength{\tabcolsep}{11pt}
	\centering
	\caption{Numerical results of each algorithm on simulated data}
	\scriptsize
	\begin{tabular}{>{\centering\arraybackslash}m{1.5cm} >{\centering\arraybackslash}m{1.5cm} l c c c}
		\toprule
		 $\kappa$&$p$ & Methods & Time(s) & ReErr &Probability\\
		  & & & &  \\
		\midrule
		\multirow{18}{*}{$0.1$}
		&\multirow{6}{*}{$1000$}
		& AFT-PDASC & 1.42e-2(2.52e-3) & 1.84e-7(3.77e-8) &1(0)\\
		&& SDAR & 2.00e-3(4.57e-4) & 1.84e-7(3.77e-8) &1(0)\\
		&& OMP & 1.77e-3(1.94e-4) & 1.84e-7(3.77e-8)&1(0) \\
		&& GreedyGP & 1.50e-3(1.16e-4) & 1.05e-6(3.54e-7) &1(0)\\
		&& AIHT & 3.40e-3(2.44e-4) & 1.93e-7(5.08e-8)&1(0) \\
		&& HTP & 1.15e-2(2.44e-3) & 1.84e-7(3.77e-8)&1(0)\\
	    \cmidrule{2-6}
	    \multirow{18}{*}
		&\multirow{6}{*}{$2000$}
		& AFT-PDASC & 4.65e-2(7.54e-3) & 1.81e-7(2.82e-8)&1(0) \\
		&& SDAR & 2.89e-3(8.04e-4) & 1.81e-7(2.82e-8)&1(0) \\
		&& OMP & 3.08e-3(2.75e-4) & 1.81e-7(2.82e-8)&1(0) \\
		&& GreedyGP & 2.95e-3(2.89e-4) & 1.13e-6(2.82e-7)&1(0) \\
		&& AIHT & 7.45e-3(1.11e-3) & 1.92e-7(4.00e-8)&1(0)  \\
		&& HTP & 2.15e-2(2.05e-3) & 1.81e-7(2.82e-8) &1(0) \\
		\cmidrule{2-6}
		\multirow{18}{*}
		&\multirow{6}{*}{$3000$}
		& AFT-PDASC & 1.36e-1(1.20e-2) & 1.88e-7(2.76e-8)&1(0) \\
		&& SDAR & 4.60e-3(1.04e-3) & 1.88e-7(2.76e-8)&1(0) \\
		&& OMP & 4.82e-3(7.69e-4) & 1.88e-7(2.76e-8)&1(0) \\
		&& GreedyGP & 5.48e-3(8.02e-4) & 1.02e-6(3.08e-7)&1(0)  \\
		&& AIHT & 1.42e-2(2.12e-3) & 2.03e-7(3.44e-8) &1(0) \\
		&& HTP & 3.08e-2(2.58e-3) & 1.88e-7(2.76e-8)&1(0)  \\	
	    \midrule
		\multirow{18}{*}{$0.4$}
		&\multirow{6}{*}{$1000$}
		& AFT-PDASC & 1.48e-2(1.91e-3) & 1.73e-7(4.05e-8)&1(0) \\
		&& SDAR & 2.31e-3(8.16e-4) & 1.73e-7(4.05e-8)&1(0) \\
		&& OMP & 2.13e-3(6.53e-4) & 1.73e-7(4.05e-8) &1(0)\\
		&& GreedyGP & 1.85e-3(4.99e-4) & 1.04e-6(3.64e-7)&1(0) \\
		&& AIHT & 7.73e-3(1.31e-3) & 4.09e-4(1.29e-3)&0.9(0.32)  \\
		&& HTP & 1.37e-2(2.34e-3) & 1.73e-7(4.05e-8)&1(0) \\
		\cmidrule{2-6}
		\multirow{18}{*}
		&\multirow{6}{*}{$2000$}
		& AFT-PDASC & 5.49e-2(1.36e-2) & 1.60e-7(2.85e-8) &1(0)\\
		&& SDAR & 3.26e-3(8.01e-4) & 1.60e-7(2.85e-8)&1(0) \\
		&& OMP & 3.02e-3(3.16e-4) & 1.60e-7(2.85e-8) &1(0)\\
		&& GreedyGP & 3.83e-3(1.19e-3) & 1.84e-5(5.48e-5)&1(0)  \\
		&& AIHT & 1.55e-2(2.22e-3) & 1.79e-7(3.72e-8)&1(0)  \\
		&& HTP & 2.58e-2(2.28e-3) & 1.60e-7(2.85e-8) &1(0) \\
		\cmidrule{2-6}
		\multirow{18}{*}
		&\multirow{6}{*}{$3000$}
		& AFT-PDASC & 1.60e-1(2.18e-2) & 1.56e-7(2.82e-8) &1(0)\\
		&& SDAR & 6.09e-3(1.71e-3) & 1.56e-7(2.82e-8)&1(0) \\
		&& OMP & 7.16e-3(2.10e-3) & 1.56e-7(2.82e-8)&1(0) \\
		&& GreedyGP & 8.58e-3(1.86e-3) & 1.12e-6(2.90e-7) &1(0) \\
		&& AIHT & 3.68e-2(1.01e-2) & 1.65e-7(3.06e-8)&1(0)  \\
		&& HTP & 4.38e-2(8.31e-3) & 1.56e-7(2.82e-8)&1(0)  \\
		\midrule
		\multirow{18}{*}{$0.7$}
		&\multirow{6}{*}{$1000$}
		& AFT-PDASC & 1.73e-2(6.72e-3) & 1.44e-7(3.50e-8)&1(0) \\
		&& SDAR & 4.36e-3(5.71e-3) & 1.44e-7(3.50e-8)&1(0)  \\
		&& OMP & 3.65e-3(3.14e-3) & 1.67e-7(7.98e-8)&0.8(0.42)  \\
		&& GreedyGP & 2.10e-3(1.28e-3) & 1.34e-2(3.31e-2) &0.7(0.48) \\
		&& AIHT & 1.03e-2(2.62e-3) & 1.06e-3(3.36e-3)&0.9(0.32) \\
		&& HTP & 1.41e-2(2.23e-3) & 6.75e-4(2.14e-3) &0.9(0.32) \\
		\cmidrule{2-6}
		\multirow{18}{*}
		&\multirow{6}{*}{$2000$}
		& AFT-PDASC & 5.31e-2(9.53e-3) & 1.33e-7(2.61e-8) &1(0)\\
		&& SDAR & 3.56e-3(5.51e-4) & 2.11e-4(6.66e-4) &0.9(0.32)\\
		&& OMP & 3.27e-3(5.58e-4) & 1.42e-7(4.05e-8)&0.7(0.48) \\
		&& GreedyGP & 2.82e-3(4.84e-4) & 1.42e-3(3.38e-3)&0.8(0.42)  \\
		&& AIHT & 2.09e-2(3.03e-3) & 2.41e-4(7.61e-4)&0.9(0.32)  \\
		&& HTP & 2.76e-2(4.20e-3) & 2.11e-4(6.66e-4)&0.9(0.32)  \\
		\cmidrule{2-6}
		\multirow{18}{*}
		&\multirow{6}{*}{$3000$}
		& AFT-PDASC & 1.31e-1(1.11e-2) & 1.25e-7(2.46e-8) &1(0)\\
		&& SDAR & 5.82e-3(1.36e-3) & 2.12e-4(6.69e-4) &0.9(0.32)\\
		&& OMP & 4.80e-3(8.19e-4) & 1.29e-7(2.52e-8)&0.8(0.42) \\
		&& GreedyGP & 5.81e-3(1.49e-3) & 7.86e-3(2.00e-2)  &0.7(0.48) \\
		&& AIHT & 3.23e-2(7.30e-3) & 4.68e-4(1.48e-3)  &0.9(0.32) \\
		&& HTP & 3.71e-2(5.08e-3) & 1.25e-7(2.46e-8) &1(0)  \\
	\bottomrule
	\end{tabular}
\label{tab1}
\end{table}

From the results in Table \ref{tab1}, it is evident that when $\kappa=0.1$ and $\kappa=0.4$, i.e., when the correlation between covariates is relatively low, all six algorithms can quickly achieve good estimation results and perform accurate recovery across the three dimensions.
However, when $\kappa=0.7$, i.e., when the correlation between covariates is high, AFT-PDASC outperforms the other algorithms in terms of relative error and recovery probability, consistently achieving precise estimation.
It is worth noting that the AFT-PDASC algorithm employs a grid search method for its regularization parameter, which results in higher computing time compared to other algorithms. However, the overall time required remains relatively low and within an acceptable range. In summary, the AFT-PDASC algorithm exhibits both high computational efficiency and numerical robustness in the simulation tests.
%%%%%%%%%%%%%%%%%%%%%%%%%%%%%%%%
\subsubsection{Performance comparisons on real-world datasets}
%%%%%%%%%%%%%%%%%%%%%%%%%%%%%%%%%
In this section, we perform numerical tests using the NKI70 breast cancer dataset, which is available in the R package. This dataset includes data from 144 breast cancer patients with lymph node-positive metastasis-free survival, along with gene expression measurements for $70$ genes that were identified as prognostic for metastasis-free survival in a previous study. The censoring rate for this dataset is $66.67\%$.
For the AFT-PDASC algorithm, we set parameters as $N=100$ and $J_{\max}=1$. Additionally, we compare the performance of AFT-PDASC with other algorithms mentioned earlier, including SDAR, OMP, GreedyGP, AIHT, and HTP. The computational results are summarized in Table \ref{tab2}.

\begin{table}
	\setlength{\tabcolsep}{5pt} %
	\renewcommand{\arraystretch}{1.0} %
	\centering
	\caption{Numerical results of each algorithm on real data `nki70'}
	\scriptsize
	\begin{tabular}{c c c c c c c c}
		\toprule
		Gene name & Number & AFT-PDASC & SDAR & OMP & GreedyGP & AIHT & HTP\\
		\midrule
		TSPYL5 & 1 & - & - & - & - &-0.38 &-\\
		Contig63649\_RC & 2 & - & - & -1.66 & -0.53 & - & -\\
		AA555029\_RC & 5 & - & - & -1.28 & - & - & -\\
		ALDH4A1 & 6 & -2.61 & - & -2.20 & -1.70 & -2.64 & -2.00\\
		Contig32125\_RC & 10 & - & - & -0.97 & - & - & -\\
		SCUBE2 & 15 & -0.41 & - & - & - & - & -\\
		EXT1 & 16 & 4.18 & 1.67 & 6.55 & 3.95 & 3.74 & 3.36\\
		GNAZ & 18 & - & 0.55 & - & - & - & 0.81\\
		MMP9 & 20 & -3.91 & -2.52 & -3.99 & -3.90 & -2.71 & -3.59\\
		RUNDC1 & 21 & - & -1.09 & -1.97 & - & - & -\\
		GMPS & 24 & - & - & -2.02 & - & - & -\\
		KNTC2 & 25 & - & - & 0.09 & 1.89 & - & -\\
		WISP1 & 26 & - & - & -2.72 & - & - & -\\
		CDC42BPA & 27 & 2.25 & - & 4.15 & 2.26 & - & 1.62\\
		GSTM3 & 30 & - & -1.19 & -1.50 & - & -1.26 & -0.81\\
		GPR180 & 31 & - & 1.02 & - & -& - & - \\
		RAB6B & 32 & - & - & -0.82 & - & - & -\\
		MTDH & 37 & -1.46 & -1.23 & - & - & - & -1.21\\
		DCK & 44 & - & - & -2.72 & -1.67 & - & -\\
		SLC2A3 & 47 & 1.49 & 1.66 & 2.69 & 2.41 & - & 1.71\\
		CDCA7 & 51 & - & -0.64 & - & - & -0.86 & -\\
		MS4A7 & 53 & - & - & 0.79 & - & - & -\\
		MCM6 & 54 & - & 1.81 & - & - & - & -\\
		AP2B1 & 55 & - & 0.96 & - & - & - & -\\
		PALM2.AKAP2 & 62 & - & - & - & - & 1.84 & -\\
		LGP2 & 63 & - & 1.30 & 1.85 & 0.97 & - & -\\
		CENPA & 66 & -2.05 & -2.09 & -1.45 & -1.97 & -0.85 & -1.66\\
		NM\_004702 & 68 & - & - & - & - & -0.41 & -\\
		ESM1 & 69 & - & - & 1.21 & - & - & -\\
		C20orf46 & 70 & - & -0.95 & - & - & -0.81 & -0.78\\
		\bottomrule
	\end{tabular}
\label{tab2}
\end{table}

From Table \ref{tab2}, we observe that the AFT-PDASC algorithm selects the fewest number of genes in most cases, while the OMP method seemly selects the most.
Moreover, the estimated coefficients for the genes chosen by all six algorithms share the same mathematics sign.
Specifically, for the 6th gene, `ALDH4A1', the regression coefficients estimated by AFT-PDASC and AIHT are quite similar. For the 20th gene, `MMP9', the coefficient estimated by AFT-PDASC is close to those obtained by OMP and GreedyGP. In the case of the 27th gene, `CDC42BPA', AFT-PDASC's coefficient is comparable to that of GreedyGP. Lastly, for the 66th gene, `CENPA', the estimated coefficient from AFT-PDASC is similar to those from SDAR and GreedyGP.
These observations suggest that the AFT-PDASC algorithm not only selects fewer genes for the NKI70 dataset but also exhibits superior estimation performance.

\section{Conclusions}\label{con}

In this paper, we addressed the AFT problem with right-censored survival data by employing a weighted least-squares method with an $\ell_0$-penalty for parameter estimation and variable selection.
For practical implementations, we developed an efficient primal dual active set algorithm and utilize a continuous strategy to select the appropriate regularization parameter.
From a theoretical perspective, we provided an error analysis for the estimated coefficients, grounded in certain assumptions regarding the covariate matrix and noise level. Additionally, we proved that the algorithm converges to the oracle solution in a finite number of steps by demonstrating the special monotonicity of the active set throughout the iterative process.
We also conducted extensive tests of the AFT-PDASC algorithm on both simulated and real-world datasets, showing its superior performance in comparison to other leading algorithms. Thus, we conclude that the AFT-PDASC algorithm is an effective tool for analyzing high-dimensional censored survival data.

\section*{Acknowledgements}
We would like to thank professor Xiliang Lu from  Wuhan university for his guidance and significant assistance in the theoretical analysis of this paper.
%We also would like to thank the anonymous referees and the associate editor for their useful comments and suggestions which improved this paper greatly.

\section*{Disclosure statement}
The authors report there are no competing interests to declare. All authors contributed to the study conception and design. All authors read and approved the final manuscript.

\section*{Funding}
%{\color{red}This work was supported by the National Natural Science Foundation of China under Grant [number 11971149].}
The work of P. Li is supported by the National Natural Science Foundation of China (Grant No. 12301420).
The work of Y. Ding is supported by the Shenzhen Polytechnic University Research Fund (Grant No. 6024310021K).
The work of Y. Xiao is supported by the National Natural Science Foundation of China (Grant No. 12471307 and 12271217), the Natural Science Foundation of Henan Province (Grant No. 232300421018).

\bibliography{references}  %参考文献库的名字Ref

\begin{thebibliography}{10}

\bibitem{B2012}
T.~Blumensath.
\newblock Accelerated iterative hard thresholding.
\newblock {\em Signal Processing}, 92(3):752--756, 2012.

\bibitem{BD2008}
T.~Blumensath and M.~E. Davies.
\newblock Gradient pursuits.
\newblock {\em IEEE Transactions on Signal Processing}, 56(6):2370--2382, 2008.

\bibitem{BJ1979}
J.~Buckley and I.~James.
\newblock Linear regression with censored data.
\newblock {\em Biometrika}, 66(3):429--436, 1979.

\bibitem{CHT2009}
T.~Cai, J.~Huang, and L.~Tian.
\newblock Regularized estimation for the accelerated failure time model.
\newblock {\em Biometrics}, 65(2):394--404, 2009.

\bibitem{CFHJZ2022}
C.~Cheng, X.~Feng, J.~Huang, Y.~Jiao, and S.~Zhang.
\newblock $\ell_0$-regularized high-dimensional accelerated failure time model.
\newblock {\em Computational Statistics $\&$ Data Analysis}, 170:107430, 2022.

\bibitem{DM2009}
W.~Dai and O.~Milenkovic.
\newblock Subspace pursuit for compressive sensing signal reconstruction.
\newblock {\em IEEE Transactions on Information Theory}, 55(5):2230--2249,
  2009.

\bibitem{FL2001}
J.~Fan and R.~Li.
\newblock Variable selection via nonconcave penalized likelihood and its oracle
  properties.
\newblock {\em Journal of the American Statistical Association},
  96(456):1348--1360, 2001.

\bibitem{F2011}
S.~Foucart.
\newblock Hard thresholding pursuit: an algorithm for compressive sensing.
\newblock {\em SIAM Journal on numerical analysis}, 49(6):2543--2563, 2011.

\bibitem{FF1993}
L.~E. Frank and J.~H. Friedman.
\newblock A statistical view of some chemometrics regression tools.
\newblock {\em Technometrics}, 35(2):109--135, 1993.

\bibitem{2002primal}
M.~Hinterm{\"u}ller, K.~Ito, and K.~Kunisch.
\newblock The primal-dual active set strategy as a semismooth \text{N}ewton
  method.
\newblock {\em SIAM Journal on Optimization}, 13(3):865--888, 2002.

\bibitem{HZ2010}
D.~Hong and F.~Zhang.
\newblock Weighted elastic net model for mass spectrometry imaging processing.
\newblock {\em Mathematical Modelling of Natural Phenomena}, 5(3):115--133,
  2010.

\bibitem{HC2013}
J.~Hu and H.~Chai.
\newblock Adjusted regularized estimation in the accelerated failure time model
  with high dimensional covariates.
\newblock {\em Journal of Multivariate Analysis}, 122:96--114, 2013.

\bibitem{JJLR2013}
J.~Huang, Y.~Jiao, B.~Jin, J.~Liu, X.~Lu, and C.~Yang.
\newblock A unified primal dual active set algorithm for nonconvex sparse
  recovery.
\newblock {\em Statistical Science}, 36(2):215--238, 2021.

\bibitem{HJLL2018}
J.~Huang, Y.~Jiao, Y.~Liu, and X.~Lu.
\newblock A constructive approach to $\ell_0$ penalized regression.
\newblock {\em Journal of Machine Learning Research}, 19(10):1--37, 2018.

\bibitem{HM2010}
J.~Huang and S.~Ma.
\newblock Variable selection in the accelerated failure time model via the
  bridge method.
\newblock {\em Lifetime Data Analysis}, 16(2):176--195, 2010.

\bibitem{HMX2006}
J.~Huang, S.~Ma, and H.~Xie.
\newblock Regularized estimation in the accelerated failure time model with
  high-dimensional covariates.
\newblock {\em Biometrics}, 62(3):813--820, 2006.

\bibitem{JJL2015}
Y.~Jiao, B.~Jin, and X.~Lu.
\newblock A primal dual active set with continuation algorithm for the
  $\ell_0$-regularized optimization problem.
\newblock {\em Applied and Computational Harmonic Analysis}, 39(3):400--426,
  2015.

\bibitem{J2008}
B.~A. Johnson.
\newblock Variable selection in semiparametric linear regression with censored
  data.
\newblock {\em Journal of the Royal Statistical Society Series B: Statistical
  Methodology}, 70(2):351--370, 2008.

\bibitem{JLZ2008}
B.~A. Johnson, D.~Lin, and D.~Zeng.
\newblock Penalized estimating functions and variable selection in
  semiparametric regression models.
\newblock {\em Journal of the American Statistical Association},
  103(482):672--680, 2008.

\bibitem{KS2016}
M.~H.~R. Khan and J.~E.~H. Shaw.
\newblock Variable selection for survival data with a class of adaptive elastic
  net techniques.
\newblock {\em Statistics and Computing}, 26(3):725--741, 2016.

\bibitem{LJLK2022}
P.~Li, Y.~Jiao, X.~Lu, and L.~Kang.
\newblock A data-driven line search rule for support recovery in
  high-dimensional data analysis.
\newblock {\em Computational Statistics $\&$ Data Analysis}, 174:107524, 2022.

\bibitem{NT2009}
D.~Needell and J.~Tropp.
\newblock Co\text{S}a\text{MP}: \text{I}terative signal recovery from
  incomplete and inaccurate samples.
\newblock {\em Applied and Computational Harmonic Analysis}, 26(3):301--321,
  2009.

\bibitem{PRK1993}
Y.~C. Pati, R.~Rezaiifar, and P.~S. Krishnaprasad.
\newblock Orthogonal matching pursuit: \text{R}ecursive function approximation
  with applications to wavelet decomposition.
\newblock In {\em Proceedings of 27th Asilomar conference on signals, systems
  and computers}, pages 40--44. IEEE, 1993.

\bibitem{S1996}
W.~Stute.
\newblock Distributional convergence under random censorship when covariables
  are present.
\newblock {\em Scandinavian journal of statistics}, pages 461--471, 1996.

\bibitem{SW1993}
W.~Stute and J.-L. Wang.
\newblock The strong law under random censorship.
\newblock {\em The Annals of statistics}, 21(3):1591--1607, 1993.

\bibitem{T1996}
R.~Tibshirani.
\newblock Regression shrinkage and selection via the lasso.
\newblock {\em Journal of the Royal Statistical Society Series B: Statistical
  Methodology}, 58(1):267--288, 1996.

\bibitem{Y1993}
Z.~Ying.
\newblock A large sample study of rank estimation for censored regression data.
\newblock {\em The Annals of Statistics}, pages 76--99, 1993.

\bibitem{Z2010}
C.-H. Zhang.
\newblock Nearly unbiased variable selection under minimax concave penalty.
\newblock {\em The Annals of Statistics}, 38(2):894--942, 2010.

\bibitem{Z2006}
H.~Zou.
\newblock The adaptive lasso and its oracle properties.
\newblock {\em Journal of the American Statistical Association},
  101(476):1418--1429, 2006.

\bibitem{ZZ2009}
H.~Zou and H.~H. Zhang.
\newblock On the adaptive elastic-net with a diverging number of parameters.
\newblock {\em The Annals of Statistics}, 37(4):1733--1751, 2009.

\end{thebibliography}
\bibliographystyle{abbrv}

\end{document}